\documentclass[11pt]{article}

\usepackage{amsfonts, amssymb, latexsym, amsmath, amsthm, mathrsfs,
            verbatim, geometry, mathtools}
\usepackage{graphics}
\usepackage{slashed}
\usepackage{url,color}
\usepackage{enumerate}

\usepackage{hyperref}
\hypersetup{
    colorlinks,
    citecolor=red,
    filecolor=green,
    linkcolor=blue,
    urlcolor=black
}

\begin{document}

\def\R{\mathbb R}
\def\N{\mathbb N}

\def\pp{\left\langle p\right\rangle}
\def\pv{\left\langle v\right\rangle}

\def\supp{\mathrm{supp}\,}
\def\div{\mathrm{div}\,}

\def\ekin{E_\mathrm{kin}}
\def\epot{E_\mathrm{pot}}
\def\A{{\cal A}}
\def\B{{\cal B}}
\def\Cc{{\cal C}}
\def\H{{\cal H}}
\def\Hc{{\cal H}_C}
\def\Hr{{{\cal H}_r}}
\def\F{{\cal F}}
\def\K{{\cal K}}
\def\L{{\cal L}}
\def\M{{\cal M}}
\def\Ro{{{\cal R}}}
\def\T{{\cal T}}

\def\Ds{{\mathscr D}}
\def\Rs{{\mathscr R}}
\def\Ns{{\mathscr N}}

\sloppy
\newtheorem{maintheorem}{Theorem}
\newtheorem{theorem}{Theorem}[section]
\newtheorem{definition}[theorem]{Definition}
\newtheorem{proposition}[theorem]{Proposition}
\newtheorem{example}[theorem]{Example}
\newtheorem{corollary}[theorem]{Corollary}
\newtheorem{lemma}[theorem]{Lemma}
\theoremstyle{remark}
\newtheorem*{remark}{Remark}

\renewcommand{\theequation}{\arabic{section}.\arabic{equation}}

\title{Stability and instability results for equilibria
  of a (relativistic) self-gravitating collisionless gas---A review}
\author{Gerhard Rein\\
        Fakult\"at f\"ur Mathematik, Physik und Informatik\\
        Universit\"at Bayreuth\\
        D-95440 Bayreuth, Germany\\
        email: gerhard.rein@uni-bayreuth.de}

\maketitle

\begin{abstract}
  We review stability and instability results for self-gravitating
  matter distributions, where the matter model is a collisionless gas
  as described by the Vlasov equation. The focus is on the general
  relativistic situation, i.e., on steady states of the Einstein-Vlasov
  system and their stability properties. In order to put things into
  perspective we include the Vlasov-Poisson system and the relativistic
  Vlasov-Poisson system into the discussion.
\end{abstract}

\tableofcontents
\section{Introduction}
\setcounter{equation}{0}
Consider a large ensemble of massive particles which interact only through the
gravitational field which they create collectively. The density $f\geq 0$ on
phase space of such a collisionless gas obeys the collisionless Boltzmann
or Liouville equation, which in mathematics is usually (and regrettably)
called the Vlasov equation. The exact form of this equation depends on
the situation at hand---Newtonian, special relativistic, or general
relativistic---, but its content is always that $f$ is constant along
particle trajectories. This equation is coupled to the field equation for
gravity which in the Newtonian case results in the Vlasov-Poisson system
and in the general relativistic one in the Einstein-Vlasov system;
these systems will be formulated in the next section.

The former system has a long history in the astrophysics literature
where it is used to model galaxies and globular clusters,
and we refer to \cite{BiTr,FrPo1984} and the references there.
The importance of the
latter system is two-fold: On the one hand there are major open, conceptual
problems in general relativity such as the cosmic censorship hypotheses
for which the choice of a matter model which by itself is well understood
is important, and the Vlasov equation is a natural candidate in this context.
On the other hand, general relativistic effects become
increasingly important in astrophysics,
given for example the fact that most galaxies have
a massive black hole at their center. Historically, this interest started
in the mid 1960s with the discovery of
quasars~\cite{BKZ,BKThorne1970,ZePo}. 

From a mathematical point of view one gains a better understanding of these
systems if one includes the so-called relativistic Vlasov-Poisson system,
a hybrid system which is neither Galilei nor Lorentz invariant.
All three systems under consideration share the property that they
have a plethora of steady state solutions which represent possible
equilibrium configurations of a galaxy or a globular cluster;
only steady states with finite mass will be relevant here.
A natural question both from
the mathematical and the astrophysical point of view is which of these
are stable, and how stable or unstable equilibria react, at least
qualitatively, to perturbations.

For the Vlasov-Poisson system several different approaches
to the stability question exist in the by now quite extensive
mathematical literature on this subject, part of which we will
recall later on. Our focus in these notes is on the stability problem for the
Einstein-Vlasov system where much less is known. We aim to
bring out the differences between the non-relativistic and the relativistic
situation and to discuss why some approaches which were successful in the
Newtonian case seem to fail in the general relativistic case and which
approaches are at least partially successful also in the relativistic case.
In the
context of general relativity the stability
question differs in an essential and striking way from its Newtonian
counterpart: In the latter case one can formulate conditions on the
so-called microscopic equation of state  which guarantee
nonlinear stability of any steady state with that
microscopic equation of state, but in the former context such a microscopic
equation of state guarantees stability only if the steady state
is not too relativistic, while a sufficiently relativistic steady state
with the same microscopic equation of state will be unstable.
This instability of strongly relativistic steady states
has no analogue in the Newtonian context.
We should also emphasize that so far the general relativistic case has only
been attacked under the assumption of spherical symmetry.

These notes are intended to be such that they can be followed
and the main ideas can be appreciated without consulting the original
literature. We aim to introduce the necessary concepts and major results
self-consistently and explain at least the ideas of most proofs.
But while some hopefully instructive proofs are given in detail,
we very often have to refer the reader to the original literature
for a complete,
rigorous analysis. Although we try to do justice to the mathematical literature
on the subject, the selection and presentation of the material
is without doubt strongly influenced by
the author's preferences, prejudices, and limitations; the coverage of the
relevant astrophysics literature is certainly quite incomplete.

The paper proceeds as follows. In the next section we state the three systems
under discussion---the Vlasov-Poisson system, the relativistic Vlasov-Poisson
system, and the Einstein-Vlasov system---together with their conserved
quantities which play
a key role in the stability analysis; we will often
use the abbreviations (VP), (RVP), (EV) to refer to these systems.
We also point out an important,
basic difference between (VP), (RVP), and (EV), according to which
they may be dubbed energy ``subcritical'', ``critical'', or ``supercritical''.
In Section~\ref{stst} we review the basic construction of one-parameter
families of steady states which in the general relativistic
case are parameterized by their central
redshift and share the same microscopic equation of state.
Section~\ref{stabstrat} recalls the basic strategies which have lead
to stability results for (VP) or (RVP), which we distinguish into
global variational methods, local variational methods, and linearization.
Section~\ref{stabkappasmall} is devoted to a linear stability result for
steady states of (EV) with small central redshift,
while Section~\ref{stabkappalarge} discusses a linear, exponential
instability result for large central redshift.
The spectral properties
of the linearized (EV) system are reviewed in Section~\ref{FA},
where we in particular discuss a recently derived Birman-Schwinger principle
for (EV). In Section~\ref{numetc} we review the main numerical
observations concerning stability for (EV) and discuss some related
conjectures and open problems. The last section provides an example
which shows that for infinite dimensional dynamical systems
strict global energy minimizers need not be stable.
\subsection*{Notation}
Since these notes are fairly long it may be useful to provide a place where
some general notation is collected which is used throughout these notes;
some of it will be re-introduced again later.

For vectors like $x,p,v \in \R^3$ we use $|\cdot|$ and $\cdot$ for the
Euclidean norm and scalar product,
\[
x\cdot v = \sum_{j=1}^3 x_j v_j,\ |x| = \sqrt{x\cdot x},
\]
etc. Gradients with respect to, say, $x$ or $p$ are denoted by
$\partial_x$ or $\partial_p$, and in order not to be too consistent
we occasionally write $\nabla$ instead of $\partial_x$. We also abbreviate
\[
\pv = \sqrt{1+|v|^2},\ w = \frac{x\cdot v}{|x|},\ L= |x\times v|^2\
\mbox{for}\ x,v \in \R^3; 
\]
this will make more sense when it first comes up. If $H$ is some Hilbert
space and $\L$ a linear, bounded or unbounded operator on $H$ we
denote by $\Ds(\L)$, $\Rs(\L)$, and $\Ns(\L)$ its domain of definition,
its range, and its null-space or kernel, i.e.,
\[
\L\colon H\supset \Ds(\L)\to H,\ \Rs(\L) = \L(\Ds(\L)),\
\Ns(\L) = \L^{-1}(\{0\})\subset \Ds(\L).
\]
\section{The systems under consideration and their conserved quantities}
\setcounter{equation}{0}
\subsection{The Vlasov-Poisson system}\label{ss_vp}
In the Newtonian case the density $f=f(t,x,p)\geq 0$ of the particle
ensemble on phase space is a function of time $t$, position
$x\in {\mathbb R^3}$, and
momentum $p\in {\mathbb R^3}$. It obeys the Vlasov-Poisson system
\begin{equation} \label{vlasov_vp}
  \partial_t f+ p\cdot \partial_x f - \partial_x U\cdot\partial_p f = 0,
\end{equation}
\begin{equation} \label{poisson}
  \Delta U = 4\pi \rho,\quad \lim_{|x|\to\infty} U(t,x) = 0,
\end{equation}
\begin{equation} \label{rho_vp}
  \rho(t,x) = \int f(t,x,p)\,dp,
\end{equation}
where $U =U(t,x)$
is the gravitational potential induced by the macroscopic, spatial mass
density $\rho=\rho(t,x)$; integrals without explicitly specified domain
extend over $\R^3$. The boundary condition in \eqref{poisson} corresponds to
the fact that we consider an isolated system in an otherwise empty universe.
As usual, we assume that all the particles in the ensemble have the same
mass which is normalized to unity so that $p$ is also the velocity of
a particle with coordinates $(x,p)$. Up to regularity issues
a function $f$ satisfies the Vlasov equation \eqref{vlasov_vp}, iff
it is constant along solutions of the equations of motion of
a test particle in the potential $U$, namely
\begin{equation} \label{chars_vp}
  \dot x = p,\ \dot p = -\partial_x U(s,x);
\end{equation}
the latter is the characteristic system of \eqref{vlasov_vp}.
If the sign in the Poisson equation is reversed, the system models a plasma,
where one will typically add a neutralizing ion background and/or an exterior
confining field.

Smooth, compactly supported initial data
$f_{|t=0} =\mathring{f} \in C^1_c(\R^6)$
launch classical, smooth solutions to this system, which are
known to be global in time \cite{LP,Pf,Sch1}, see also the review
\cite{Rein07}. These solutions conserve
the following
quantities, which we define as functionals acting on states $f=f(x,p)$:
\begin{equation} \label{energy_vp}
  \H(f) \coloneqq \ekin (f) + \epot (f)
  \coloneqq \frac{1}{2} \iint |p|^2 f(x,p)\,dp\,dx
  - \frac{1}{8 \pi} \int |\nabla U_f(x)|^2 dx
\end{equation}
is the total energy of the state $f$, i.e., the sum of its kinetic and
potential energies, where the potential $U_f$ is induced by $f$ via
\eqref{poisson} and \eqref{rho_vp}, and
\begin{equation} \label{casi_vp}
  \Cc(f) \coloneqq \iint \Phi(f(x,p))\,dp\,dx
\end{equation}
is a so-called Casimir functional, which is conserved for any choice of
$\Phi \in C^1([0,\infty[)$ with $\Phi(0)=0$.
The fact that the energy $\H$ is conserved along solutions
of the Vlasov-Poisson system simply says that the latter is a conservative
system, while the conservation of the Casimir functionals corresponds
to the fact that the characteristic flow induced by \eqref{chars_vp}
preserves Lebesgue measure. In other words, $f(t)$,
the state of the system at time $t$, is related to $\mathring{f}$ via
\begin{equation} \label{ftfn}
  f(t) = \mathring{f} \circ Z(0,t)
\end{equation}
where
\[
s \mapsto (X,P)(s,t,x,p) = Z(s,t,z)
\]
is the solution to \eqref{chars_vp} with $(X,P)(t,t,x,p)= (x,p)$,
which induces a diffeomorphism
\[
Z(t,0) = Z(t,0,\cdot)\colon \R^6 \to \R^6
\]
with inverse $Z(0,t)$, and
\[
\det \partial_z Z (t,0) = 1.
\]
Both types of conservation laws are essential
for deducing global-in-time existence of solutions and for nonlinear
stability issues, cf.\ Section~\ref{ss_criticality}.

Before we proceed to relativistic models we mention a different way of writing
the Vlasov equation \eqref{vlasov_vp}. To this end we recall the Poisson
bracket of two smooth functions $g=g(x,p)$ and $h=h(x,p)$,
\begin{equation}\label{pbracket}
  \{g,h\} \coloneqq \partial_x g \cdot \partial_p h
  - \partial_p g \cdot \partial_x h,
\end{equation}
and the energy of a particle with coordinates $(t,x,p)$,
\begin{equation}\label{parten_vp}
  E = E(t,x,p) = \frac{1}{2} |p|^2 + U(t,x).
\end{equation}
Then the Vlasov equation \eqref{vlasov_vp} can be written as
\begin{equation}\label{vlasov_bracket}
  \partial_t f + \{f,E\} = 0 .
\end{equation}
We recall that $\cdot$ 
denotes the Euclidean scalar
product between vectors in $\R^3$, and the Euclidean norm of such vectors
is denoted by $|\cdot|$.
\subsection{The relativistic Vlasov-Poisson system}\label{ss_rvp}
For this system the Vlasov equation takes the form
\begin{equation} \label{vlasov_rvp}
  \partial_t f+ \frac{p}{\sqrt{1+|p|^2}}\cdot \partial_x f -
  \partial_x U\cdot\partial_p f = 0,
\end{equation}
where we again assume that all the particles have the same rest mass,
normalized to unity, and the speed of light is set to unity as well.
The Poisson equation \eqref{poisson} together with its boundary condition
and the relation \eqref{rho_vp} remain unchanged.
The characteristic system now reads
\[ 
\dot x = \frac{p}{\sqrt{1+|p|^2}},\ \dot p = -\partial_x U(s,x),
\]
and the relation \eqref{ftfn} remains true with the flow map
redefined accordingly. The characteristic flow is still measure preserving
so that we keep the Casimir functionals \eqref{casi_vp} as conserved
quantities, and (RVP) is still conservative with the obvious
change that now
\[ 
\ekin (f) 
\coloneqq \iint \sqrt{1+|p|^2} f(x,p)\,dp\,dx.
\]
The Vlasov equation \eqref{vlasov_rvp} can again be put into the form
\eqref{vlasov_bracket} with \eqref{parten_vp} replaced by
\begin{equation}\label{parten_rvp}
  E = E(t,x,p) = \sqrt{1+|p|^2} + U(t,x).
\end{equation}
As mentioned above, this system is neither Galilei nor Lorentz
invariant. While it may not be so relevant from the physics point of view
it will be useful in illustrating the difficulties which the stability
discussion encounters when moving from (VP) to (EV).
Initial data as specified for (VP) launch local, classical, smooth solutions
of (RVP) which can easily be seen by adapting the proof of
\cite[Thm.~1.1]{Rein07}. But it is known that such solutions
can blow up in finite time, cf.~\cite{GlSch}. In Section~\ref{ss_criticality}
we will explain this difference to (VP) and consider the question what
this means with respect to stability.

\subsection{The Einstein-Vlasov system}
\label{ss_ev}
On a smooth spacetime manifold $M$
equipped with a  Lorentzian metric $g_{\alpha \beta}$
with signature $(-{}+{}+{}+)$
the Einstein equations read
\begin{equation} \label{feqgen}
  G_{\alpha \beta} = 8 \pi T_{\alpha \beta}.
\end{equation}
Here $G_{\alpha \beta}$ is the Einstein tensor induced by the metric, and
$T_{\alpha \beta}$ is the energy-momentum tensor; Greek indices run from
$0$ to $3$.
The world line of a test particle
on $M$ obeys the geodesic equation,
which can be written either as a first order ODE on the tangent bundle
$TM$ of the spacetime manifold, coordinatized by $(x^\alpha, p^\beta)$
where $x^\alpha$ are general coordinates on $M$ and
$p^\alpha$ are the corresponding canonical momenta,
or on the cotangent bundle
$TM^\ast$, coordinatized by $(x^\alpha, p_\beta)$
where $p_\beta= g_{\beta\gamma} p^\gamma$. If we opt for the latter alternative,
\[
\dot x^\alpha = g^{\alpha\beta} p_\beta,\
\dot p_\alpha = - \frac{1}{2}\partial_{x^\alpha} g^{\beta \gamma} p_\beta p_\gamma;
\]
$g^{\alpha \beta}$ denotes the inverse of the metric $g_{\alpha \beta}$,
the dot indicates differentiation with respect to proper time along the
world line of the particle, and the Einstein summation convention is applied.
All the particles are to have the same
rest mass which we normalize to unity, and to move forward in time.
Their number density $f$ is a non-negative function
supported on the mass shell
\[
PM^\ast \coloneqq \left\{ g^{\alpha \beta} p_\alpha p_\beta = -1,\ p^\alpha \
\mbox{future pointing} \right\},
\]
a submanifold of the cotangent bundle $TM^\ast$
which is invariant under the geodesic flow.
Letting Latin indices range from $1$ to $3$
we use coordinates $(t,x^a)$ with zero shift which implies that
$g_{0a}=0$. On the mass shell $PM^\ast$ the variables $p_0$ and $p^0$
then become functions of the variables $(t,x^a,p_b)$:
\[
p_0 = - |g_{00}|^{1/2} \sqrt{1+g^{ab}p_a p_b},\
p^0 =  |g_{00}|^{-1/2}  \sqrt{1+g^{ab}p_a p_b}.
\]
Since the number density
$f=f(t,x^a,p_b)$ is constant along the geodesics,
the Vlasov equation reads
\begin{equation} \label{vlgen}
  \partial_t f + \frac{g^{a b} p_b}{p^0}\,\partial_{x^a} f
  -\frac{1}{2 p^0} \partial_{x^a} g^{\beta \gamma} p_\beta p_\gamma \partial_{p_a} f
  = 0.
\end{equation}
The energy-momentum tensor is given as
\begin{equation} \label{emtvlgen}
  T_{\alpha \beta}
  = |g|^{-1/2}\int p_\alpha p_\beta f\,\frac{dp_1 dp_2 dp_3}{p^0},
\end{equation}
where $|g|$ denotes the modulus of the determinant of the metric.
The system \eqref{feqgen}, \eqref{vlgen}, \eqref{emtvlgen}
is the Einstein-Vlasov system in general coordinates. As we
want to describe isolated systems, we require that the
spacetime is asymptotically flat which corresponds to the boundary condition
in \eqref{poisson}.

An obvious steady state of this system is flat Minkowski space with $f=0$.
In \cite{FaJoSm,LiTa19,taylor} nonlinear stability of this trivial
steady state was shown for the system above, which is a highly non-trivial
result.
Under the simplifying assumption of spherical symmetry this result
was shown in \cite{Rein95,RR92a}. Mathematically speaking, these results
are small data results which rely on the fact that close to vacuum
the characteristic flow of the Vlasov equation disperses the
matter in space. When perturbing a non-trivial, i.e., non-vacuum, steady state
no such mechanism exists, and the problem becomes completely
different. Our discussion is focused exclusively on the stability
of non-trivial steady states.

Questions like the stability or instability of
non-trivial steady states are at present out of reach
of a rigorous mathematical treatment, unless simplifying symmetry
assumptions are made. We assume spherical symmetry, use
Schwarzschild coordinates $(t,r,\theta,\varphi)$, and write
the metric in the form
\begin{equation} \label{metric}
  ds^2=-e^{2\mu(t,r)}dt^2+e^{2\lambda(t,r)}dr^2+
  r^2(d\theta^2+\sin^2\theta\,d\varphi^2).
\end{equation}
Here $t\in \R$ is a time coordinate,
and the polar angles $\theta\in [0,\pi]$ and $\varphi\in [0,2\pi]$
coordinatize the surfaces of constant $t$ and $r>0$.
The latter are the orbits
of $\mathrm {SO}(3)$, which acts isometrically on this spacetime, and
$4 \pi r^2$ is the area of these surfaces. The boundary condition
\begin{equation}\label{boundcinf}
  \lim_{r\to\infty}\lambda(t, r)=\lim_{r\to\infty}\mu(t, r)=0
\end{equation}
guarantees asymptotic flatness, and in order to guarantee a regular center
we impose the boundary condition
\begin{equation}\label{boundc0}
  \lambda(t,0)=0.
\end{equation}
Polar coordinates have a tendency to introduce artificial singularities at
the center. Hence it is convenient to also use the corresponding Cartesian
coordinates
\[
x = (x^1,x^2,x^3) =
r (\sin \theta \cos \varphi,\sin \theta \sin \varphi,
\cos \theta)
\]
and the corresponding canonical covariant momenta $p=(p_1,p_2,p_3)$.

Before we proceed to formulate (EV) in these variables we emphasize
the fact that from this point on we will not raise or lower any indices,
treat $x$ and $p$ simply as variables in $\R^3$, and use notations
like
\[
x\cdot p = \sum_{a=1}^3 x^a p_a,\ |p|^2 = \sum_{a=1}^3 (p_a)^2
\]
for Euclidean scalar products and norms, just as we did for (VP) or (RVP).

In order that the particle distribution function
$f=f(t,x,p)$ is compatible with \eqref{metric} it
must be spherically symmetric; we call a state $f=f(x,p)$
{\em spherically symmetric} iff
\begin{equation} \label{sphsydef}
  f(x,p) = f(Ax,Ap),\ x, p \in \R^3,\ A \in \mathrm{SO}\,(3).
\end{equation}
Using the abbreviation
\begin{equation} \label{p0p}
  \pp \coloneqq - e^{-\mu}p_0
  = \sqrt{1+|p|^2 +(e^{2\lambda}-1)\left(\frac{x\cdot p}{r}\right)^2},
\end{equation}
(EV) can be put into the following form:
\begin{align} \label{vlasov_evp}
  \partial_t f
  &+ e^{\mu - 2\lambda}\frac{p}{\pp}\cdot \partial_x f \nonumber\\
  &+
  \left[e^{\mu - 2\lambda}\lambda' \left(\frac{x\cdot p}{r}\right)^2\frac{1}{\pp}
    - e^{\mu} \mu' \pp + e^\mu\frac{1-e^{-2\lambda}}{r\pp}
    \left(|p|^2 - \left(\frac{x\cdot p}{r}\right)^2\right) \right]
  \frac{x}{r} \cdot \partial_p f =0,
\end{align}
\begin{equation}
  e^{-2\lambda} (2 r \lambda' -1) +1
  =
  8\pi r^2 \rho , \label{ein1}
\end{equation}
\begin{equation}
  e^{-2\lambda} (2 r \mu' +1) -1
  =
  8\pi r^2 \sigma, \label{ein2}
\end{equation}
\begin{equation}
  \dot\lambda =
  - 4 \pi r e^{\lambda + \mu} \jmath, \label{ein3}
\end{equation}
\begin{equation}
  e^{- 2 \lambda} \left(\mu'' + (\mu' - \lambda')
  (\mu' + \frac{1}{r})\right)
  - e^{-2\mu}\left(\ddot\lambda +
  \dot\lambda \, (\dot\lambda - \dot\mu)\right)
  =
  8 \pi \sigma_T, \label{ein4}
\end{equation}
where
\begin{align}
  \rho(t,r)
  &=
  \rho(t,x) = e^{-\lambda} \int \pp f(t,x,p)\,dp ,\label{rp}\\
  \sigma(t,r)
  &=
  \sigma(t,x) = e^{-3 \lambda} \int \left(\frac{x\cdot p}{r}\right)^2
  f(t,x,p)\frac{dp}{\pp}, \label{pp}\\
  \jmath(t,r)
  &=
  \jmath(t,x) = e^{-2 \lambda} \int \frac{x\cdot p}{r} f(t,x,p) dp, \label{jp}\\
  \sigma_T(t,r)
  &=
  \sigma_T(t,x) = \frac{1}{2} e^{-3 \lambda}
  \int \left|{\frac{x\times p}{r}}\right|^2
  f(t,x,p) \frac{dp}{\pp}. \label{pp_T}
\end{align}
Here $\dot{}$ and ${}'$ denote the derivatives with respect to
$t$ and $r$ respectively, $\rho$ is the mass-energy density---its integral
is the ADM mass, cf.~\eqref{energy_evp}---, and $\sigma$, $\sigma_T$
are the pressure in the radial or tangential direction, respectively.

The equations \eqref{vlasov_evp}--\eqref{pp_T} are a form of the spherically
symmetric (EV) which does not look too appealing and has so far not been used
in the literature. The fact that the source terms defined in
\eqref{rp}--\eqref{pp_T} depend on the metric,
in particular via \eqref{p0p}, makes it technically unpleasant to handle.
But this form of the system has some advantages. The characteristic flow
of the Vlasov equation \eqref{vlasov_evp} is again measure preserving,
and hence the Casimir functionals defined exactly as in \eqref{casi_vp}
remain conserved quantities.
Moreover, the Vlasov equation \eqref{vlasov_evp} still is of the general form
\eqref{vlasov_bracket} with
\[ 
E = E(t,x,p) = e^\mu \pp.
\]
The total energy, which in this case is usually referred to as the ADM mass,
is given as
\begin{equation} \label{energy_evp}
  \H(f)
  \coloneqq \iint e^{-\lambda_f}
  \sqrt{1+|p|^2 +(e^{2\lambda_f}-1)\left(\frac{x\cdot p}{r}\right)^2}
  f(x,p)\,dp\,dx
\end{equation}
where $\lambda_f$ is the solution to \eqref{ein1} subject to the boundary
conditions from \eqref{boundcinf} and \eqref{boundc0} and with $\rho$
satisfying \eqref{rp}.

We rewrite the above form of (EV) by introducing 
non-canonical momentum variables via
\begin{equation} \label{vdef}
  v = p + (e^\lambda -1)\frac{x\cdot p}{r} \, \frac{x}{r}.
\end{equation}
In these variables \eqref{p0p} turns into
\begin{equation} \label{p0v}
  \pv \coloneqq - e^{-\mu} p_0
  = \sqrt{1+|v|^2},
\end{equation}
in the definition of spherical symmetry of $f=f(t,x,v)$ we simply replace
$p$ by $v$, and the Vlasov equation \eqref{vlasov_evp} becomes
\begin{equation} \label{vlasov_evv}
  \partial_t f + e^{\mu - \lambda}\frac{v}{\pv}\cdot \partial_x f -
  \left( \dot \lambda \frac{x\cdot v}{r} + e^{\mu - \lambda} \mu'
  \pv \right) \frac{x}{r} \cdot \partial_v f =0 .
\end{equation}
The field equations \eqref{ein1}--\eqref{ein4} remain unchanged, but
the source terms
\begin{align}
  \rho(t,r)
  &=
  \rho(t,x) = \int \pv f(t,x,v)\,dv ,\label{rv}\\
  \sigma(t,r)
  &=
  \sigma(t,x) = \int \left(\frac{x\cdot v}{r}\right)^2
  f(t,x,v)\frac{dv}{\pv}, \label{pv}\\
  \jmath(t,r)
  &=
  \jmath(t,x) = \int \frac{x\cdot v}{r} f(t,x,v) dv, \label{jv}\\
  \sigma_T(t,r)
  &=
  \sigma_T(t,x) = \frac{1}{2} \int \left|{\frac{x\times v}{r}}\right|^2
  f(t,x,v) \frac{dv}{\pv}. \label{pv_T}
\end{align}
are now given completely in terms of $f$, they do no longer depend on
the metric. The price to pay for this simplification is that the
characteristic flow of the Vlasov equation \eqref{vlasov_evv} is not
measure preserving, and the Casimir functionals, which are still
conserved quantities, take the form
\begin{equation} \label{casi_evv}
  \Cc(f) \coloneqq \iint e^{\lambda_f} \Phi(f(x,v))\,dv\,dx.
\end{equation}
On the other hand, the ADM mass simplifies to a linear functional
that depends only on $f$,
\begin{equation} \label{energy_evv}
  \H(f)
  \coloneqq \iint \pv f(x,v)\,dv\,dx = \iint \sqrt{1+|v|^2} f(x,v)\,dv\,dx.
\end{equation}
Taking into account the boundary conditions
\eqref{boundcinf} and\eqref{boundc0}
the metric components are given explicitly in terms of $\rho$ and $\sigma$,
and hence of the state $f=f(x,v)$; we suppress the time variable $t$
for the moment:
\begin{equation} \label{lambdam}
  e^{-2\lambda} = 1-\frac{2 m}{r}
\end{equation}
and
\begin{equation} \label{muprm}
  \mu' = e^{2\lambda}\left(\frac{m}{r^2} + 4\pi r \sigma\right), 
\end{equation}
where 
\begin{equation} \label{mdef}
  m(r) = 4\pi \int_0^r \rho(s)\, s^2 ds.
\end{equation}
At this point we notice that a spacetime manifold can only be covered by
Schwarzschild coordinates if $2 m < r$ everywhere; the spacetime must not
contain trapped surfaces. Finally, we also mention that the structure
\eqref{vlasov_bracket} is lost when using the non-canonical momentum
variable $v$. As in most of the stability-related literature
we will use the version of (EV) in non-canonical variables,
but since many important aspects of the stability issue are still widely
open (even in spherical symmetry), it may be useful to keep the
alternative, canonical formulation in mind. It is also possible that
other coordinates adapted to spherical symmetry are more suitable
for the stability analysis. We will not pursue this issue but mention
maximal areal coordinates as one alternative \cite{GueRe21}.
\subsection{A basic difference between (VP), (RVP), and (EV)}
\label{ss_criticality}
Let us suppose that we want to make use of conservation of energy to get insight
into global existence issues for the initial value problem or stability issues.
Then in the case of (VP) or (RVP) we must deal with the fact that
while $\ekin + \epot$ is conserved, the two terms have opposite signs, and
no immediate control of  $\ekin$ or  $\epot$ results.

In what follows we sometimes employ the notation
\[
\rho_f(x)\coloneqq\int f(x,p)\, dp
\]
for the spatial density induced by some measurable phase-space
density $f=f(x,p)\geq 0$. Similarly, we will write $U_\rho$ or
$U_f$ for the potential induced by $\rho$ or $\rho_f$ via \eqref{poisson}.

Now let $0\leq k\leq \infty$ and
$n=k+ 3/2$. For any $R>0$,
\begin{align*}
  \rho_f (x)
  &=
  \int_{|p|\leq R} f(x,p)\, dp + \int_{|p| > R} f(x,p)\, dp\\
  &\leq
  \left(\frac{4\pi}{3}R^3\right)^\frac{1}{k+1} \|f(x,\cdot)\|_{1+1/k}
  + \frac{1}{R^2} \int |p|^2 f(x,p)\, dp,
\end{align*}
where $\|\cdot\|_s$ denotes the usual $L^s$ norm, in this case over $\R^3$.
We choose
\[
R=\left(\int |p|^2 f\, dp/\|f(x,\cdot)\|_{1+1/k}\right)^{\frac{1+k}{5+2 k}},
\]
take the resulting estimate to the power $1+1/n$ and integrate
with respect to $x$ to conclude that
\begin{equation}\label{rhoestvp}
  \|\rho_f\|_{1+1/n} \leq C  \|f\|_{1+1/k}^{\frac{2 + 2 k}{5 + 2 k}}
  \left(\iint |p|^2 f(x,p)\, dp\, dx\right)^{\frac{3}{5 + 2 k}}.
\end{equation}
On the other hand, the Hardy-Littlewood-Sobolev inequality 
\cite[4.3 Thm.]{LiLo} implies that
\[
-\epot(f) \leq C \|\rho_f\|_{6/5}^2.
\]
We require that
\[
1+\frac{1}{n} \geq \frac{6}{5},\ \mbox{i.e.},\ n \leq 5,\ \mbox{i.e.},
k \leq \frac{7}{2}.
\]
Then we can interpolate the $L^{6/5}$ norm between the $L^1$ and $L^{1+1/n}$
norms and conclude that
\begin{equation} \label{keyvp}
  -\epot(f) \leq C \|f\|_1^{\frac{7-2k}{6}}\|f\|_{1+1/k}^{\frac{k+1}{3}} \ekin(f)^{1/2}.
\end{equation}
This key estimate has important consequences for (VP). Assume that we have
a local-in-time, smooth solution to this system, which conserves the total
energy and both $\|f(t)\|_1$ and $\|f(t)\|_\infty$.
Then \eqref{keyvp} with $k=0$ implies that along this solution both
$\epot(f(t))$ and $\ekin(f(t))$ and hence also $\|\rho(t)\|_{5/3}$ remain
bounded, which are key a-priori bounds towards global-in-time existence.

Staying with (VP) we now suppose that
we want to minimize an energy-Casimir functional
\[
\H(f) + \Cc(f)
\]
under the constraint that the mass $\iint f=M$ is prescribed
and for a Casimir function $\Phi$ which grows sufficiently fast
to control $\|f\|_{1+1/k}$ with $k\leq 7/2$;
a corresponding minimizer will be a candidate for a stable steady state,
cf.~Section~\ref{stabstrat}. Then the key estimate \eqref{keyvp} implies
that along a corresponding minimizing sequence $\ekin(f)$ and hence also
$\|\rho_f\|_{1+1/n}$ remain bounded, which is an important step
towards a necessary compactness argument along such a minimizing sequence.
We see that the success of both global-in-time existence results and stability
results via global variational techniques hinges on the estimate \eqref{keyvp}.

Let us check how \eqref{keyvp} fares in the (RVP) case.
Proceeding as before,
\[
\|\rho_f\|_{\frac{k+4}{k+3}} \leq C  \|f\|_{1+1/k}^{\frac{k+1}{k+4}}
\left(\iint \pp f(x,p)\, dp\, dx\right)^{\frac{3}{k+4}},
\]
and \eqref{keyvp} turns into
\begin{equation} \label{keyrvp}
  -\epot(f) \leq C \|f\|_1^{\frac{2-k}{3}}\|f\|_{1+1/k}^{\frac{k+1}{3}} \ekin(f),\
  0\leq k\leq 2.
\end{equation}
This estimate gives no control
on $\ekin$ along either a local-in-time solution or a minimizing sequence
for the variational problem mentioned above.
The point here is that for (RVP) the kinetic energy is
only a first order moment in $p$ while for (VP) it is a second order moment.

For (EV) the situation is even worse in the following sense.
The key point above is that the kinetic energy is a higher-order moment
in $p$ than what appears in the definition of $\rho$ so that
some  $L^s$ norm of $\rho$ with $s>1$ is under control, provided
$\ekin$ is under control. But as we see from the formula for the energy
\eqref{energy_evv} in the (EV) situation, this energy gives us exactly an
$L^1$ bound for $\rho$ and nothing more.
This missing ``something more'' makes (EV) that much harder to deal with,
both with respect to global-in-time existence and with respect to stability.
\section{Steady states}
\label{stst}
\setcounter{equation}{0}
Before we address their stability we must recall what typical steady states
of (VP), (RVP), or (EV) look like, and how one can establish their existence.
To this purpose let us suppose that we are given a time-independent
potential $U=U(x)$ or a time-independent metric of the form \eqref{metric}.
Then the particle energy, defined for (VP) or (RVP) in \eqref{parten_vp} or
\eqref{parten_rvp} and for the non-canonical form of (EV) as
\begin{equation}\label{parten_evv}
  E=E(x,v) = e^{\mu}\pv,
\end{equation}
is constant along characteristics of the corresponding static Vlasov equation
and hence solves that equation. For (VP) or (RVP) we therefore make the ansatz
\begin{equation}\label{ststansatz_vp}
  f(x,p) = \phi(E) = \varphi(E_0 -E)
\end{equation}
for the particle distribution function, which for technical reasons we modify
to
\begin{equation}\label{ststansatz_ev}
  f(x,v) = \phi(E) = \varphi\left(1 -\frac{E}{E_0}\right)
\end{equation}
in the (EV) case. We refer to the relation \eqref{ststansatz_vp}
or \eqref{ststansatz_ev} as a {\em microscopic equation of state}.
To keep matters simple we assume that
\begin{equation}\label{phicond0}
  \varphi\in C(\R) \cap C^2(]0,\infty[),\quad
  \varphi=0\  \mbox{on}\ ]-\infty,0],\quad
  \varphi>0\ \mbox{on}\ ]0,\infty[;
\end{equation}
$E_0$ is a cut-off energy with $E_0<0$ for (VP)
or (RVP) and $0<E_0<1$ for (EV). Such a cut-off energy is necessary to
obtain steady states with a localized matter distribution.
Notice that the ansatz function $\phi$ depends on the
cut-off energy $E_0$, which must be specified to in order to specify $\phi$,
but the ansatz function $\varphi$ does not depend on $E_0$, and it has the
reversed monotonicity behavior with respect to $E$.

With this ansatz we satisfy the Vlasov equation,
the source terms become functionals of $U$ or $\mu$, respectively, and
the static systems reduce the the field equation(s) with this dependence
substituted in. In the (EV) case these functions are spherically symmetric
by assumption, but also in the (VP) and (RVP) case the ansatz
\eqref{ststansatz_vp} leads to steady states which necessarily are
spherically symmetric, cf.~\cite{GNN}. In particular,
both $U$ and $\mu$ can be viewed as functions of $r=|x|$.
Instead of looking for $U$ or $\mu$ directly, we
define a new unknown
\[ 
y(r) = E_0 -U(r)\ \mbox{or}\ y(r) = E_0 -U(r)-1
\]
in the (VP) or (RVP) case, respectively, and
\[ 
y(r) = \ln E_0 -\mu(r)
\]
in the (EV) case. In the latter case,
\begin{equation} \label{rhoyrel}
  \rho(r) = g(y(r)), \quad
  \sigma(r) = h(y(r)) = \sigma_T(r),
\end{equation}
where
\begin{equation} \label{gdef}
  g(y) \coloneqq 
  4 \pi e^{4 y} \int_0^{1-e^{-y}} \varphi(\eta)\, (1-\eta)^2\,
  \left((1-\eta)^2-e^{-2y}\right)^{1/2} d\eta
\end{equation}
and
\begin{equation} \label{hdef}
  h(y) \coloneqq 
  \frac{4 \pi}{3} e^{4 y} \int_0^{1-e^{-y}} \varphi(\eta)\,
  \left((1-\eta)^2-e^{-2y}\right)^{3/2} d\eta.
\end{equation}
The functions $g$ and $h$ are continuously differentiable 
on $\R$, cf.\ \cite[Lemma 2.2]{RR00}, they are strictly decreasing for $y>0$,
and they vanish for $y<0$.
For (VP),
\[ 
\rho(r) = g_N (y(r)),\ \mbox{where}\
g_N (y) \coloneqq 
4 \pi \sqrt{2} \int_0^{y} \varphi(\eta)\,
\left(y-\eta\right)^{1/2} d\eta;
\]
the subscript $N$ stands for ``Newtonian'', and the exact form of
the analogous relation for (RVP) is not relevant here. We recall that
$\lambda$ is given in terms of $\rho$ via \eqref{lambdam}, and the static
(EV) system is reduced to 
\begin{equation} \label{yeq}
  y'(r)= - \frac{1}{1-2 m(r)/r} \left(\frac{m(r)}{r^2}
  + 4 \pi r \sigma(r)\right) ,\
  y(0)=\kappa
\end{equation}
cf.~\eqref{muprm};
here $m$, $\rho$, and $\sigma$ are given in terms of $y$ by \eqref{rhoyrel}
and \eqref{mdef}, and $\kappa>0$ is prescribed.
The static (VP) or (RVP) systems reduce to
\begin{equation} \label{yeqvp}
  y'(r)= - \frac{m(r)}{r^2},\ y(0)=\kappa.
\end{equation}
For any given $\kappa >0$ a fixed point argument yields a unique, smooth,
local solution to \eqref{yeq} or \eqref{yeqvp} on some short interval
$[0,\delta]$. The solution $y$ is strictly decreasing, can be extended
to exist on $[0,\infty[$, and either remains strictly positive, or has
a unique zero at some radius $R>0$ beyond which there is vacuum.
The crucial question is for which ansatz functions
$\phi$ respectively $\varphi$ the latter case holds, because in that case
the above procedure yields steady states with compact support and finite mass.
Once such a solution $y$ is obtained, $E_0\coloneqq\lim_{r\to\infty}y(r)$
and $U(r)\coloneqq E_0-y(r)$ defines the cut-off
energy and the potential in the (VP) or (RVP) case, while
$E_0\coloneqq\exp(\lim_{r\to\infty}y(r))$ and
$\mu(r)=\ln E_0 - y(r)$ for (EV); in either case the boundary condition
at infinity follows. 
For more details to these arguments we refer to
\cite{RaRe} and the references there.

A sufficient condition on $\varphi$ which guarantees finite mass and compact
support of the resulting steady states in all three cases,
(VP), (RVP) and (EV), is that
\begin{equation} \label{phicond1}
  \varphi(\eta) \geq C \eta^k \ \mbox{for}\ \eta \in ]0,\eta_0[
\end{equation}
for some parameters $C>0$, $\eta_0>0$, and $0<k<3/2$, cf.~\cite{RaRe};
in passing we note that conditions on $\varphi$ or $\phi$ which are both
necessary and sufficient for finite radius and finite mass are
not known. To sum up:
\begin{proposition} \label{ssfamilies}
  Let $\varphi$ satisfy \eqref{phicond0} and \eqref{phicond1}.
  \begin{itemize}
  \item[(a)]
    There exists a one-parameter family of steady states
     $(f_{\kappa},U_{\kappa})_{\kappa >0}$
    of the spherically symmetric (VP) (or (RVP)) system,
    and $\kappa = U_\kappa(R_\kappa) - U_\kappa(0)$.
  \item[(b)]
    There exists a one-parameter family of steady states
     $(f_{\kappa},\lambda_{\kappa},\mu_{\kappa})_{\kappa >0}$
     of the spherically symmetric, asymptotically flat
     (EV) system, and
     $\kappa = \mu_\kappa(R_\kappa) - \mu_\kappa(0)$.
  \end{itemize}
  The spatial support of such a steady state is an interval
  $[0,R_\kappa]$ with $0<R_\kappa<\infty$,
  $\rho_\kappa, \sigma_\kappa \in C^1([0,\infty[)$,
  $y_\kappa, U_\kappa, \mu_\kappa, \lambda_\kappa  \in C^2([0,\infty[)$, and
  $\rho'_\kappa(0)=\sigma'_\kappa(0)= y'_\kappa(0) = U_\kappa'(0) = \mu'_\kappa(0)
  = \lambda'_\kappa(0)=0$. Moreover, we denote
  $D=D_\kappa \coloneqq \{ f_\kappa >0\}$ so that $\supp f_\kappa = \overline{D_\kappa}$,
  which is compact in $\R^6$.
\end{proposition}
An essential difference between the (VP) and the (EV) case concerning
the stability of steady states is the following: For (VP) one can formulate
conditions on the microscopic equations of state---$\varphi$
in \eqref{ststansatz_vp} should be strictly increasing on
$[0,\infty[$---which guarantee
that all steady states in the corresponding one-parameter family from
Proposition~\ref{ssfamilies} are nonlinearly stable; this remains true
even for the King model $\varphi(\eta)=(e^\eta-1)_+$ where the mass-radius
diagram, mentioned in item (d) of the remark below,
exhibits a spiral structure, cf.~\cite{GuRe2007}.
For (EV) the same type of microscopic equation of state
will yield a one-parameter family where the individual steady states
change from being stable to being unstable as the central redshift
$\kappa$ increases from small values to larger ones. There is by now ample
numerical evidence for this behavior \cite{AnRe1,Gue_e_a,GueStRe21},
and we will discuss the first steps towards an analytic
understanding of this behavior. To do so, we must understand the
consequences which very small or very large values of $\kappa$
have on the structure of the corresponding steady states
$(f_{\kappa},\lambda_{\kappa},\mu_{\kappa})_{\kappa >0}$ in the (EV) case.

In what follows we make the dependence of the (EV) steady states
on $\kappa$ explicit in the notation only when we study the limits
$\kappa\to 0$ in Section~\ref{stst_nonrel} and $\kappa\to\infty$
in Section~\ref{stst_ultrarel} or when the logic of some statement
requires this.
For other parts of our discussion, in particular for the
(VP) case, the value of $\kappa$ plays no role or is fixed,
and we will abuse
notation in saying that $(f_0,U_0)$ or $(f_0,\lambda_0,\mu_0)$ is a steady
state of (VP) or (EV), which is to be understood in the sense
that some $\kappa_0>0$ is fixed and $f_0\coloneqq f_{\kappa_0}$ etc.

One should also notice that many other quantities depend on
$\kappa$ such as the particle energy
\[
E = E(x,v)=e^{\mu_\kappa}\pv,
\]
the set $D= \{ f_\kappa >0\}$, the ansatz function $\phi$ in
\eqref{ststansatz_ev} via the cut-off energy $E_0$, and various operators
introduced in Sections~\ref{stabkappasmall},~\ref{stabkappalarge},~\ref{FA}.
These dependencies on $\kappa$ will always be suppressed in our notation.

We conclude our steady state discussion with some remarks.
\begin{remark}
  \begin{itemize}
  \item[(a)]
    The steady states obtained in Proposition~\ref{ssfamilies} are
    isotropic in the sense
    that $\sigma = \sigma_T$; we use $\sigma$ to denote the (radial)
    pressure also in the Newtonian case.
    In the (EV) case they satisfy the following identities
    on  $[0,\infty[$, the second of which is
    known as the Tolman-Oppenheimer-Volkov equation:
    \begin{equation} \label{laplusmu}
      \lambda_\kappa' + \mu_\kappa'
      = 4 \pi r e^{2\lambda_\kappa} \left(\rho_\kappa +\sigma_\kappa\right),
    \end{equation}
    \begin{equation} \label{tov}
      \sigma_\kappa' = - \left(\rho_\kappa +\sigma_\kappa\right) \,\mu_\kappa'.
    \end{equation}
  \item[(b)]
    A remarkable property of these steady states is that their induced
    macroscopic quantities solve the Euler-Poisson or Einstein-Euler
    system respectively. Given the fact that the functions $g_N$ or $g$
    are one-to-one for $y>0$ one can write $y$ as a function of $\rho$,
    and substituting into the relation for the pressure in \eqref{rhoyrel}
    yields the corresponding macroscopic equation
    of state $\sigma = \sigma(\rho)$, which is part of the corresponding
    Einstein-Euler of Euler-Poisson system.
  \item[(c)]
    The parameter $\kappa$ which parameterizes the above steady state families
    is the difference in the potential between the center and the boundary
    of the matter distribution. In the (EV) case it is related
    to the redshift factor $z$
    of a photon which is emitted at the center $r=0$ and received at the
    boundary $R_\kappa$ of the steady state; this is not the standard definition
    of the central redshift where the photon is received at infinity,
    but it is a more suitable parameter here:
    \[ 
    z = \frac{e^{\mu_\kappa(R_\kappa)}}{e^{\mu_\kappa(0)}} - 1
    = \frac{e^{y_\kappa(0)}}{e^{y_\kappa(R_\kappa)}} - 1 = e^\kappa -1.
    \]
    Although this is not the standard terminology
    we refer to $\kappa$ as the central redshift, and we will see later
    that it is a measure for how non-relativistic or relativistic a steady
    state is. In the (VP) case the parameter $\kappa$ seems to have no
    effect on the stability
    properties of the corresponding steady states, but in the (EV) case steady
    states with sufficiently large $\kappa$ will be seen to be unstable.
  \item[(d)]
    An instructive way to visualize one of these one-parameter families
    of steady states is to plot, for a certain parameter range, the points
    $(M_\kappa,R_\kappa)$ representing the (ADM) mass and radius of the state
    with parameter $\kappa$. The resulting curve is referred to as a
    mass-radius curve. For (VP) these curves can be strictly
    monotonic, for example in the polytropic case
    $\varphi(\eta)=\eta_+^k$, or they
    can exhibit a spiral structure, for example for the King model
    $\varphi(\eta)=(e^\eta-1)_+$, cf.~\cite{RaRe2017}. In strong contrast, these
    curves always have a spiral structure in the (EV) case,
    cf.~\cite{AnRe2,Mak}.
    This is interesting, because according to the so-called turning
    point principle \cite{So1981} passing through a turning point on
    the mass-radius spiral should affect the stability behavior of the
    steady state. For the Einstein-Euler system a rigorous version of
    this principle has been proven in \cite{HaLin}, see also
    \cite{HaLinRe}, but the principle does not hold in the (EV) case,
    cf.~\cite{Gue_e_a,GueStRe21}. The principle is known to be false for
    (VP) where for example all the steady states along the mass-radius
    spiral for the King model are known to be nonlinearly stable.
  \item[(e)]
    Due to spherical symmetry the quantity
    \[
    L\coloneqq |x\times p|^2,
    \]
    the modulus of angular momentum squared, is conserved along
    characteristics of both (VP) and (EV); for the latter system,
    $L=|x\times v|^2$. Hence one may include a
    dependence on $L$ into the microscopic equation of state
    \eqref{ststansatz_vp} or \eqref{ststansatz_ev}.
    Resulting steady states are then no longer isotropic,
    i.e., $\sigma \neq \sigma_T$, and the correspondence to the Euler
    matter model explained in part (b) above is lost. A common way
    to include the $L$-dependence is to generalize \eqref{ststansatz_ev} to
    \begin{equation}\label{ststansatz_evL}
      f(x,v) = \phi(E,L) = \varphi\left(1 -\frac{E}{E_0}\right) (L-L_0)_+^l .
    \end{equation}
    Here $l> -1/2$, and the analogous ansatz is used for (VP).
    If $L$ is bounded away from zero, i.e., $L_0>0$, then the resulting
    steady states have a vacuum region at the center, if $L_0=0$ they do not.
    The static shell solutions with $L_0>0$ look somewhat artificial, but they
    become more interesting if one places a Schwarzschild black hole
    (or a point mass in the (VP) case) into the vacuum region, which is
    then surrounded by a static shell of Vlasov matter,
    cf.~\cite{GueStrRe,Jab,Rein94,Rein99b}.
  \end{itemize}
\end{remark}
\section{Strategies towards stability in the (VP) and (RVP) case}
\label{stabstrat}
\setcounter{equation}{0}
In this section we recall the main methods which have resulted in stability
results for the Vlasov-Poisson or the relativistic Vlasov-Poisson system.
We do not aim for completeness, but only wish to give some orientation
on what approaches one may try for the stability problem in the
Einstein-Vlasov case. Our discussion will be even less complete concerning
results from the astrophysics literature. All the available results rely
explicitly or implicitly on the condition that the ansatz
\eqref{ststansatz_vp} or \eqref{ststansatz_evL} is strictly
decreasing in $E$ on its support:
\begin{equation}\label{mainstabcond}
  \phi'(E) < 0\ \mbox{for}\ E<E_0\ \mbox{or}\
  \partial_E \phi(E,L) < 0\ \mbox{for}\ E<E_0, L\geq L_0.
\end{equation}
\subsection{Global variational methods}
\label{globvar_vp}
Let us consider the following variational problem:
Minimize the energy-Casimir functional
\[
\Hc = \H + \Cc
\]
over the set
\[
\F_M \coloneqq\left\{ f \in L^1(\R^6)\, \mid \, f\geq 0,\
\iint f dp\,dx = M,\  \ekin(f) + \Cc(f) < \infty \right\}.
\]
Here the kinetic, potential, and total energy $\ekin$, $\epot$, and $\H$
are defined as in \eqref{energy_vp}, the Casimir functional $\Cc$ is defined
in \eqref{casi_vp} where for the moment we take
\begin{equation} \label{polyPsi}
  \Phi(f)\coloneqq\frac{k}{1+k} f^{1+1/k}
\end{equation}
with some $k\in ]0,3/2[$, and $M>0$ is fixed.
Since $k<3/2$ we can choose $\alpha$ such that
$0<\frac{\alpha}{2},\frac{k}{3}\frac{\alpha}{\alpha-1} < 1$.
The key estimate \eqref{keyvp}
and Young's inequality then imply that for any $f\in\F_M$,
\begin{align}\label{hclowerb}
  \Hc(f)
  &\geq
  \ekin(f) + \Cc(f) - C \Cc(f)^{k/3} \ekin(f)^{1/2} \notag \\
  &\geq
  \ekin(f) - C \ekin(f)^{\frac{\alpha}{2}} +
  \Cc(f) - C \Cc(f)^{\frac{k}{3}\frac{\alpha}{\alpha-1}},
\end{align}
where the constant $C>0$ depends on $M$, $k$, and $\alpha$.
This estimate implies that
\[
h_M\coloneqq\inf_{\F_M}\Hc > -\infty.
\]
Now consider a minimizing sequence $(f_j)\subset \F_M$, i.e.,
$\Hc(f_j) \to h_M$.
Then by \eqref{hclowerb}, $\ekin(f_j)$ and $\Cc(f_j)$ remain bounded,
in particular, $(f_j)$ is a bounded sequence in $L^{1+1/k}(\R^6)$ which by the
Banach-Alaoglu theorem has a weakly convergent subsequence,
again denoted by $(f_j)$.
Its limit is a natural candidate for a global minimizer of $\Hc$ over $\F_M$. 
By \eqref{rhoestvp} the sequence of induced spatial densities $(\rho_j)$
is bounded and (up to a subsequence) weakly convergent in $L^{1+1/n}(\R^3)$.

The key difficulty now is to upgrade these weak convergences in such a way
that one can pass to the limit in the potential energy; the kinetic energy
is not a problem since it is linear in $f$. More generally speaking, some
sort of compactness argument must be applied to the minimizing sequence.
The following lemma, which is proven for example in \cite[Lemma~2.5]{Rein07},
captures the compactness property of the solution operator to
the Poisson equation; recall that $U_\rho$ or $U_f$ denotes the
Newtonian gravitational potential induced by a spatial density
$\rho$ or a phase-space density $f$. 
\begin{lemma} \label{concimplcomp}
  Let $0<n<5$. Let $(\rho_j) \subset L^{1+1/n} (\R^3)$ be such that
  \begin{align} \label{conc}
    &
    0\leq \rho_j \rightharpoonup \rho_0 \
    \mbox{weakly in}\  L^{1+1/n} (\R^3), \ \mbox{and}\nonumber \\
    &
    \forall \epsilon > 0 \;\exists R>0 :\
    \limsup_{j\to \infty} \int_{|x|\geq R} \rho_j(x)\, dx < \epsilon  .
  \end{align}
  Then 
  $\nabla U_{\rho_j} \to \nabla U_{\rho_0}$ strongly in $L^2(\R^3)$.
\end{lemma}
Under our assumption on $k$ it holds that $n=k+3/2 <3$, so the key issue is to
verify \eqref{conc}, i.e., the minimizing sequence must in
essence remain concentrated.
To do this, one may employ the concentration-compactness principle introduced
by {\sc P.-L.~Lions} \cite{Lions} combined with an analysis of how $\epot(f)$
behaves under scalings and splittings, or one may rely on the latter
arguments exclusively, and all this is greatly simplified if one restricts
the discussion to spherical
symmetry; we refer to \cite{Rein07} and the references there for details.
At this point one should note that while the steady states under consideration
are spherically symmetric anyway, an a-priori restriction to spherical
symmetry in the variational problem
limits a resulting stability result to spherically symmetric
perturbations and is thus undesirable.

In the concentration argument it turns out that in order to achieve \eqref{conc}
the ball in which the mass remains concentrated must be allowed to shift
with the sequence; notice that all the functionals used above are invariant
under translations in $x$. The resulting existence result for the above
variational problem reads as follows.
\begin{theorem} \label{exminim}
  Let $(f_j) \subset \F_M$ be a minimizing sequence of 
  $\Hc$. Then there exists 
  a function $f_0\in \F_M$, a subsequence,
  again denoted by $(f_j)$ and a sequence $(a_j) \subset \R^3$ of shift vectors
  such that 
  \begin{align*}
    f_j(\cdot + a_j,\cdot) \rightharpoonup f_0 
    &
    \ \mbox{weakly in}\ L^{1+1/k} (\R^6),\ j\to \infty,\\
    \nabla  U_{f_j} (\cdot + a_j) 
    \to \nabla U_{f_0}
    &
    \ \mbox{strongly in}\ L^2 (\R^3),\ j\to \infty,
  \end{align*}
  and the state $f_0$ minimizes the energy-Casimir functional $\Hc$ over $\F_M$.
\end{theorem}
One should realize that the compactness along minimizing sequences
captured in the theorem above is indispensable for concluding that
the state $f_0$ is a nonlinearly stable steady state of (VP);
its minimizer property is not sufficient for stability.
To appreciate this point, we now discuss how stability is obtained;
a pedagogical example which further illustrates this issue,
which is typical for infinite dimensional dynamical systems,
will be given in Section~\ref{theexample}.

First we remark that by standard arguments which can for example be found
in \cite[Theorem~5.1]{Rein07} the minimizer obtained
in Theorem~\ref{exminim} is of the form
\begin{equation} \label{poly}
  f_0(x,p) = (E_0 - E)_+^k
\end{equation}
with $E$ defined as in \eqref{parten_vp} with the induced potential
$U_0=U_{f_0}$; $E_0$ arises as a Lagrange multiplier.
So $f_0$ is a polytropic steady state of (VP).

For $f \in \F_M$, 
\begin{equation} \label{expansion1}
  \Hc (f)- \Hc (f_0)=d(f,f_0)-\frac{1}{8 \pi}
  \int|\nabla U_f - \nabla U_0|^2 dx, 
\end{equation}
where 
\begin{align*}
  d(f,f_0)
  &\coloneqq
  \iint \left[\Phi(f)-\Phi(f_0) + E\, (f-f_0)\right]\,dp\,dx\\
  &= 
  \iint \left[\Phi(f)-\Phi(f_0) + (E-E_0)\, (f-f_0)\right]\,dp\,dx\\
  &\geq
  \iint \left[ \Phi'(f_0) + (E - E_0)\right]
  \,(f-f_0)\,\,dp\,dx \geq 0
\end{align*}
with $d(f,f_0)=0$ iff $f=f_0$. Let us define
\begin{equation} \label{Ddef}
  \mathrm{dist}(f,f_0) \coloneqq d(f,f_0) + \frac{1}{8 \pi}
  \int|\nabla U_f - \nabla U_0|^2 dx.
\end{equation}
Notice the switch in the sign between \eqref{expansion1} and \eqref{Ddef};
$\mathrm{dist}(f,f_0)$ is a perfectly fine measure for the distance of
a perturbation $f$ from $f_0$.
We obtain the following nonlinear stability result.
\begin{theorem}\label{stabilitygal}
  Let $f_0$ be a minimizer as obtained in Theorem~\ref{exminim}.
  Then for every $\epsilon>0$ there exists a $\delta>0$ such that for every
  classical solution $t \mapsto f(t)$ of the Vlasov-Poisson system 
  with $f(0) \in C^1_c (\R^6)\cap \F_M$ 
  the initial estimate
  \[
  \mathrm{dist}(f(0),f_0) < \delta
  \]
  implies that for every $t\geq 0$ there is a shift vector
  $a\in \R^3$ such that
  \[
  \mathrm{dist}(f(t,\cdot + a,\cdot),f_0) < \epsilon.
  \]
\end{theorem}
\begin{proof}
  Assume the assertion is false. Then there exist 
  $\epsilon>0,\ t_j>0,\ f_j(0) \in C^1_c (\R^6)\cap \F_M$ such that
  for $j\in \N$,
  \[
  \mathrm{dist}(f_j(0),f_0) < \frac{1}{j},
  \]
  but for any shift vector $a\in \R^3$,
  \[
  \mathrm{dist}(f_j(t_j,\cdot + a,\cdot),f_0) \geq \epsilon.
  \]
  Since $\Hc$ is conserved, (\ref{expansion1}) and the assumption on 
  the initial data imply that
  $\Hc (f_j(t_j)) =  \Hc (f_j(0)) \to \Hc (f_0)$,
  i.e., $(f_j(t_j)) \subset  \F_M$ is a minimizing sequence.
  Hence by Theorem~\ref{exminim},
  $\int |\nabla U_{{f_j}(t_j)}-\nabla U_0|^2\to 0$ 
  up to subsequences and shifts in $x$, provided that there is
  no other minimizer to which this sequence can converge.
  By (\ref{expansion1}), 
  $d(f_j(t_j),f_0)\to 0$ as well, which is the desired
  contradiction.

  For the polytropic case \eqref{poly} there exists for each value
  of the total mass $M$ exactly one corresponding steady state---up to
  shifts in $x$---which provides the uniqueness of the minimizer $f_0$
  used above.
\end{proof}

The spatial shifts in the above arguments are necessary due to the Galilei
invariance of the problem, and a stability assertion of the form above
is sometimes referred to as orbital stability \cite{Mou}.

At the end of this subsection we will briefly comment on various variations
and extensions of the basic theme discussed so far, but one variation deserves
some attention. As seen from Lemma~\ref{concimplcomp} the basic compactness
mechanism along minimizing sequences operates on spatial densities $\rho$.
Following \cite{Rein02} we define for $r \geq 0$,
\[ 
{\cal G}_r \coloneqq\left\{g \in L^1 (\R^3) |\ g\geq 0,\  
\int\left(\frac{1}{2}|p|^2 g (p) + \Phi(g(p))\right) dp < \infty,\
\int g(p)\, dp = r \right\} 
\]
and
\[ 
\Psi(r)\coloneqq\inf_{g \in {\cal G}_r} 
\int\left(\frac{1}{2}|p|^2 g (p) + \Phi(g(p))\right) dp.
\]
We consider the problem of minimizing the reduced functional
\begin{equation} \label{hcrdef}
  \Hr(\rho)
  \coloneqq
  \int \Psi (\rho(x))\, dx + \epot (\rho)
\end{equation}
over the set
\[ 
{\cal R}_M
\coloneqq
\left \{\rho \in L^1(\R^3) \mid \rho\geq 0,\
\int \Psi (\rho(x))\, dx < \infty ,\ \int \rho (x)\, dx = M 
\right\};
\]
the potential energy $\epot(\rho)$ is defined
in the obvious way and is finite for states in this constraint set.
For the polytropic choice \eqref{polyPsi},
\[
\Psi(r) = c_n r^{1+1/n},\ r\geq 0,
\]
which should be compared with the estimates introduced in
Section~\ref{ss_criticality};
here $n=k+3/2$ as before, and $c_n>0$ is some constant.
There is a close relation between the reduced
variational problem and the original one. 
For every function $f \in \F_M$,
\[
\Hc(f) \geq \Hr(\rho_f),
\]
and if $f=f_0$ is a minimizer of $\Hc$ over $\F_M$ then equality holds,
i.e., the reduced functional ``supports'' the original one from below.
Moreover, if $\rho_0 \in {\cal R}_M$ is a minimizer of $\Hr$ with
induced potential $U_0$ then it can be lifted to a minimizer
$f_0$ of $\Hc$ in $\F_M$ as follows: The Euler-Lagrange equation
for the reduced functional says that
\[ 
\rho_0 = \left\{ 
\begin{array}{ccl}
  (\Psi')^{-1}(E_0 - U_0)&,& U_0 < E_0, \\
  0 &,& U_0 \geq E_0, 
\end{array}
\right.
\]
where $E_0$ is the corresponding Lagrange multiplier.
With the particle energy $E$ defined as before the function
\[
f_0
\coloneqq
\left\{ 
\begin{array}{ccl}
  (\Phi')^{-1}(E_0 - E) &,& E < E_0, \\
  0 &,& E \geq E_0, 
\end{array}
\right.
\]  
is a minimizer of $\Hc$ in $\F_M$; for the details
cf.~\cite[Theorem~2.1]{Rein07}.

To attack the variational problem through the reduced functional
has several advantages. The minimizer of the reduced functional can be shown
to be a nonlinearly stable steady state of the Euler-Poisson system with
macroscopic equation of state $\sigma = \sigma(\rho) = c_n \rho^{1+1/n}$.
The relation between the latter system and (VP) which was noted for
isotropic steady states carries over to their stability properties,
cf.\ \cite{Rein03,Rein07}. More important for the present context,
compactness properties are easier to study for the reduced functional,
because the latter lives on a space of functions of $x$ and, in case
of spherical symmetry, of the $1d$ variable $r=|x|$. In addition, a result of
{\sc Burchard} and {\sc Guo} \cite[Thm.~1]{BG} shows that if one minimizes
the reduced functional only over spherically symmetric densities
$\rho=\rho(|x|)$, the resulting minimizer is actually a minimizer over the
full set ${\cal R}_M$, and the stability result Theorem~\ref{stabilitygal}
is recovered.

A somewhat different reduced functional which acts on the mass functions
$m_f(r)=4\pi \int_0^r \rho_f(s)\, s^2 ds$ induced by spherically symmetric
phase-space densities $f$ was used in \cite{Wo1}. This was historically
the first rigorous stability result for (VP), but for (VP) the approach
was not explored any further.
The approach may become useful for (EV), cf.~\cite{AnKu20,AnKu,Wo}.

As mentioned before, there are many variations to the basic theme discussed
above, and we mention some:
\begin{remark}
  \begin{itemize}
  \item[(a)]
    The form of the Casimir functional can be much more general than
    the prototypical form \eqref{polyPsi}.
    Strict convexity of $\Phi$ and growth conditions for small and for large
    arguments compatible with \eqref{polyPsi} are sufficient.
    Strict convexity of $\Phi$ corresponds to the main stability condition
    \eqref{mainstabcond}.
  \item[(b)]
    For such more general Casimir functionals the uniqueness of the minimizer,
    which played a role in the proof of Theorem~\ref{stabilitygal},
    will in general be lost, but this is not essential for the
    stability argument, cf.~\cite{Sch3}.
  \item[(c)]
    Instead of minimizing the energy-Casimir functional $\H + \Cc$ under
    the mass constraint $\iint f = M$, one can also minimize the energy $\H$
    under the mass-Casimir constraint $\iint f + \Cc(f)=M$.
    This has the advantage that
    one can cover the polytropes \eqref{poly} for $0<k<7/2$,
    cf.~\cite{GuRe2001,Rein07}, and, with some extra effort, also the limiting
    case $k=7/2$, the so-called Plummer sphere;
    for $k>7/2$ finite mass and physical relevance are lost.
  \item[(d)]
    The reduction mechanism does no longer work for the situation
    described in (c), but this is as it should be,
    since for $k>3/2$, i.e., $n>3$, stability of the corresponding Euler-Poisson
    steady states is lost, cf.~\cite{Jang}. That the (VP) steady states remain
    stable also for $k>3/2$ shows that the parallels between the Euler and
    the Vlasov matter models have their (obscure) limitations.
  \item[(e)]
    By making the Casimir functional depend on the angular momentum
    variable $L$, in which case it should no longer be called
    ``Casimir'' functional, steady states
    depending on $L$ can be covered, cf.~\cite{Guo99,GuRe99,GuRe99b}.
    Besides such spherically symmetric, non-isotropic states one can apply
    the method also to states with axial symmetry, with a point mass at
    the center, or to flat steady states with or without a dark matter halo,
    cf.~\cite{Firt07,FiRe,FiReSe,GuRe03,Rein99a,Schulze_09}.
  \item[(f)]
    One can also minimize the energy $\H$ under two separate constraints, a mass
    constraint and a Casimir constraint, cf.~\cite{SS}. Along these lines
    the arguably strongest result on global minimizers for (VP) was obtained
    by {\sc Lemou, Mehats, Rapha\"el} in \cite{LeMeRa1}.
  \end{itemize}  
\end{remark}
\subsection{Local minimizers; the structure of $D^2 \H_C$}
\label{locvar_vp}
The global minimizer approach reviewed in the previous subsection has been quite
successful, but it also has limitations. Suppose we want to investigate the
stability of the King model, an important steady state of (VP) which appears
in the astrophysics literature, obtained via
\[ 
\varphi(\eta)=(e^\eta-1)_+ .
\]
Then the function $\Phi$ in the corresponding Casimir functional
becomes
\[ 
\Phi(f)=(1+f)\ln(1+f) -f,
\]
which grows too slowly to control any $L^s$ norm of $f$ with $s>1$,
and hence the key estimate \eqref{keyvp} cannot be brought into play.
If instead we consider (RVP), then the corresponding
estimate \eqref{keyrvp} provides no control in the context of the global
variational problem to begin with,
so the method from the previous subsection fails.
Notice further that the global minimizer method provides the existence of a
steady state which then turns out to be stable. The method is not really
one for addressing the stability of some given steady state, obtained
by some other method.

In the present subsection we discuss an approach which aims to show that
a given steady state $f_0$ is a local minimizer of an energy-Casimir functional
by examining the structure of the latter near $f_0$. The method was
first applied to the King model in the (VP) context, cf.~\cite{GuRe2007}.
Following \cite{HaRe2007} we review this approach in the context of (RVP),
which is a little closer to (EV) where for analogous reasons the global
approach seems to fail as well.

We consider some fixed, isotropic steady state $(f_0,\rho_0,U_0)$ of (RVP)
given by an ansatz like \eqref{ststansatz_vp}, and 
an energy-Casimir functional defined as before. By a (formal) expansion,
\begin{align*}
  \Hc (f) = \,
  &
  \Hc (f_0) + 
  \iint (E + \Phi'(f_0))\,(f-f_0)  \,dv\,dx \nonumber \\
  & 
  {}- \frac{1}{8 \pi} \int|\nabla U_f-\nabla U_0 |^2 dx 
  + \frac{1}{2} \iint \Phi''(f_0) (f-f_0)^2 dp\,dx 
  + \ldots,
\end{align*}
and we now define
$\Phi\colon [0,\infty[ \to \R$ such that $f_0$ becomes a critical point
of $\Hc$, namely 
\[ 
\Phi(f)\coloneqq -\int_0^f\phi^{-1}(z)\,dz,\ f\in [0,\infty[,
\]
so that $\Phi\in C^2([0,\infty[)$ with $\Phi'(f_0) = - E$ on $\supp f_0$.
To simplify the discussion
we restrict ourselves to the polytropic form \eqref{poly} with
$1\leq k < 7/2$ where the above becomes rigorous;
the key assumption is again that on its support the ansatz strictly
decreases in the particle energy $E$, cf.~\eqref{mainstabcond}.
The question whether $f_0$ is a strict local minimizer of $\Hc$ obviously
depends on the behavior of the quadratic term in the expansion above,
i.e., on
\[
D^2\Hc (f_0)(g,g) \coloneqq  \frac{1}{2}\iint_{\{f_{0}>0\}}
\frac{1}{|\phi'(E)|} g^2\, dp\, dx 
-\frac{1}{8\pi}\int |\nabla U_{g}|^2 dx;  
\]
we write the argument $g$ twice to emphasize that this is a term
which is quadratic in $g$, and we notice that $\phi'<0$ where $f_0>0$.
It was a remarkable insight in the astrophysics community and for the (VP) case
that on so-called linearly dynamically accessible states
$g=\{f_0,h\}=\phi'(E)\{E,h\}$ the quadratic term
$D^2\Hc (f_0)(g,g)$ is positive definite, cf.~\cite{KS,SDLP}, and the analogous
result holds for (RVP); the Poisson bracket $\{\cdot,\cdot\}$ was introduced in
\eqref{pbracket}.
\begin{lemma}\label{kandrup_rvp}
  Let $h\in C_c^{\infty}(\mathbb R^6)$ be spherically symmetric 
  with $\supp h \subset \{f_0>0\}$ and such that $h(x,-p)=-h(x,p)$. 
  Then the following inequality holds:
  \[
  D^2\Hc(f_0)(\{E,h\},\{E,h\})
  \geq
  \frac{1}{2} \iint \frac{1}{|\phi'(E)|}
  \!\left[|x\cdot p|^2 \left|\left\{E,\frac{h}{x\cdot p}\right\}\right|^2
  \!+\frac{U_0'}{r(1+|p|^2)^{3/2}}h^2\right] dp\,dx.
  \]
\end{lemma}
This lemma is proven in \cite[Lemma~3.4]{HaRe2007}. It provides 
positive definiteness of $D^2\Hc (f_0)$
on dynamically accessible states in a quantified manner.
A crucial step in any stability analysis is to specify the set of admissible
perturbations. In astrophysical reality, perturbations arise
by some exterior force acting on the steady state ensemble.
It redistributes the particles in phase space by
a measure preserving flow, leading to perturbations of
the form $f=f_0\circ T$ with $T\colon\R^6 \to \R^6$ a measure preserving
diffeomorphism. Such perturbations are called
{\em dynamically accessible from} $f_0$. For the case at hand we 
restrict ourselves to spherically symmetric such perturbations and require
that the diffeomorphism $T\colon\R^6 \to \R^6$
{\em respects spherical symmetry}, i.e., for all $x,p \in \R^3$ and all
rotations $A\in \mathrm{SO}(3)$,
\[
T(Ax,Ap)=(Ax',Ap')\
\mbox{and}\ |x' \times p'| = |x\times p|,\ \mbox{where}\ (x',p')=T(x,p).
\]
From a physics point of view this restriction is undesirable.
The set of admissible perturbations is defined as
\begin{align*}
  {\cal D}_{f_0} \coloneqq \Bigl\{ f=f_0 \circ T \mid 
  &
  \;T\colon \R^6 \to \R^6 \
  \mbox{is a measure preserving $C^1$-diffeomorphism}\\
  &
  \;\mbox{which respects spherical symmetry}\Bigr\}\,.
\end{align*}
This set is invariant under
classical solutions of (RVP).
At least formally, states of the bracket form $g=\{f_0,h\}$
are tangent vectors to the manifold ${\cal D}_{f_0}$ at the point $f_0$, and the set
of these states is invariant under the linearized
dynamics; this terminology is borrowed from Hamiltonian dynamics,
cf.~\cite{mor}.

We are going to measure the distance of a state $f\in {\cal D}_{f_0}$ 
from the steady state $f_0$ by the same quantity which we used
in the previous subsection, namely
\[
\mathrm{dist}(f,f_0) = \iint[\Phi(f)-\Phi(f_0)+E (f-f_0)]\,dp\, dx
+\frac{1}{8\pi}\int|\nabla U_f-\nabla U_0|^2\,dx,
\]
see \eqref{Ddef}. Then
\begin{equation} \label{decrel}
  \mathrm{dist}(f,f_0)
  = \Hc(f)- \Hc(f_0)+\frac{1}{4\pi}\int|\nabla U_f-\nabla U_0|^2\,dx.
\end{equation}
It can be shown by Taylor expansion
that there exists a constant $C>0$ which depends
only on the steady state $f_0$ such that
\[ 
\|f - f_0\|_2^2 + \|\nabla U_f-\nabla U_0\|_2^2
\leq C \mathrm{dist}(f,f_0),\ f\in {\cal D}_{f_0},
\]
cf.~\cite[Lemma~3.1]{HaRe2007}.
The key result is the following theorem
which says---in a precise, quantified manner---that
the steady state is a local minimizer of the energy-Casimir functional
in the set ${\cal D}_{f_0}$.
\begin{theorem}\label{locmin_rvp}
  There exist constants $\delta_0>0$ and $C_0>0$ such that for 
  all $f\in {\cal D}_{f_0}$ with $\mathrm{dist}(f,f_0)\leq\delta_0$ the
  following estimate holds:
  \[
  \Hc(f)- \Hc(f_0)\geq C_0 \|\nabla U_f-\nabla U_0\|_2^2.
  \]
\end{theorem}
The proof goes by contradiction: If the theorem were false, one could
eventually construct a linearly dynamically accessible state that would
contradict the positive definiteness of $D^2\Hc (f_0)$ obtained in
Lemma~\ref{kandrup_rvp}; for the quite technical and non-trivial details
we refer to
\cite{HaRe2007}.  Stability of $f_0$ is an immediate corollary.
\begin{theorem}\label{nlstab_rvp}
  There exist constants $\delta >0$ and $C>0$ such that 
  every solution $t\mapsto f(t)$ of (RVP) which starts close to $f_0$
  in the sense that $f(0) \in {\cal D}_{f_0}$ with
  $\mathrm{dist}(f(0),f_{0}) < \delta$,
  exists globally in time and satisfies the estimate
  \[
  \mathrm{dist}(f(t),f_{0}) \leq C \; \mathrm{dist}(f(0),f_{0}),\ t\geq 0.
  \]
\end{theorem}
\begin{proof}
  With $\delta_0$ and $C_0$ from Theorem~\ref{locmin_rvp},
  define $\delta\coloneqq \delta_0 (1+1/(4 \pi C_0))^{-1}$,
  and consider a solution
  $[0,T[\ni t\mapsto f(t)$ of (RVP) with 
  $f(0)\in {\cal D}_{f_0}$ on some maximal interval of existence; a suitable local
  existence result can be found in \cite{KoeRe}.
  Assume that 
  \[
  \mathrm{dist}(f(0),f_{0}) < \delta < \delta_0. 
  \]
  By continuity,
  \[
  \mathrm{dist}(f(t),f_{0}) < \delta_0,\ t\in [0,t^\ast[, 
  \]
  where $0<t^\ast \leq T$ is chosen maximal.
  Since $f(t) \in {\cal D}_{f_0}$ for all $t\in [0,T[$, 
  Theorem~\ref{locmin_rvp}, the relation (\ref{decrel}),
  and the fact that $\Hc$ is conserved yield the following
  chain of estimates for 
  $t\in[0,t^\ast[$:
  \begin{align*}
    \mathrm{dist}(f(t),f_0)
    &=
    \Hc(f(t)) - \Hc(f_0) + \frac{1}{4\pi} \|\nabla U_{f(t)}-\nabla U_0\|_2^2\\
    &\leq
    \Hc(f(t)) - \Hc(f_0) +
    \frac{1}{4 \pi C_0} \left(\Hc(f(t)) - \Hc(f_0)\right)\\
    &=
    \left(1+\frac{1}{4 \pi C_0}\right) \, \left(\Hc(f(0)) - \Hc(f_0)\right)
    \leq
    \left(1+\frac{1}{4 \pi C_0}\right)\,\mathrm{dist}(f(0),f_0) < \delta_0.
  \end{align*}
  This implies that $t^\ast = T$. Thus $\ekin(f(t))$ is bounded on $[0,T[$
  which for spherically symmetric solutions 
  is sufficient to conclude that $T=\infty$, cf.~\cite[Prop.~4.1]{HaRe2007}
  and \cite{KoeRe}.
\end{proof}
A nice feature of this theorem,
also in view of the (EV) case, is that the stability estimate
provided by the theorem implies global existence of spherically
symmetric solutions which start close enough to $f_0$, while spherically
symmetric solutions of (RVP) with $\H (f(t))<0$ are known to blow up
in finite time, cf.~\cite{GlSch}.
The only previously known global solutions
of (RVP) were small data solutions (and steady states).
We should also point out that the need for the spatial shifts
which were necessary in Theorem~\ref{stabilitygal} is eliminated by the
restriction to spherical symmetry.

We close the discussion of the local minimizer approach by some comments on
possible variations and extensions.
\begin{remark}
  \begin{itemize}
  \item[(a)]
    We restricted ourselves to the polytropic case \eqref{poly} in order to
    avoid formulating the general conditions on the steady state
    which are needed for the above arguments, cf.~\cite{HaRe2007}.
  \item[(b)]
    The method works for (VP) just as well and was introduced in \cite{GuRe2007}
    to deal with the King model for which the global method fails.
  \item[(c)]
    The local minimizer method can be combined with a suitable reduction of the
    energy functional. In \cite{LeMeRa2} this was done for (VP) by
    exploiting the monotonicity of $\H$ under generalized symmetric
    rearrangements.
    This analysis was inspired by results from the astrophysics literature
    \cite{Aly1989,PerAl1996,LyBe1969,WiZiSch88} and was restricted
    to spherical symmetry. The latter restriction was removed in \cite{LeMeRa3}
    which provides arguably the strongest result on (VP) in the spirit of
    the present subsection.
  \end{itemize}
\end{remark}
\subsection{Linearization}
\label{ss_linvp}
Why is it that linearization has up to this point not shown up in this
review, when this approach is probably the first that one encounters in the
relevant mathematics courses and when it figures most prominently in the
relevant astrophysics literature~\cite{An1961,BiTr,DoFeBa,FrPo1984,KS}?
For a possible answer we have to look
at the linearization of (VP) about some given steady state $f_0$,
which we take as isotropic; $f_0=\phi(E)$.

If we substitute $f = f_0 + \delta f$ into (VP), use the fact that $f_0$
is a stationary solution, and drop the term which is quadratic in $\delta f$
with the justification that $\delta f$ is very small, the result is the equation
\begin{equation}\label{vplingen}
  \partial_t \delta f + \T  \delta f
  - \nabla U_{\delta f(t)}  \cdot p\, \phi'(E) =0, 
\end{equation}
where
\[
\T \coloneqq p\cdot\partial_x - \nabla U_0 \cdot\partial_p = \{\cdot,E\}
\]
is the transport operator associated to the steady state $f_0$, i.e., the
operator which generates the characteristic flow in the steady state potential
$U_0=U_{f_0}$, cf.~\cite{RStr}.
Following {\sc Antonov}~\cite{An1961}
we split $\delta f = \delta f_+ + \delta f_-$
into its even and odd parts with respect to $p$,
\[ 
\delta f_\pm (t,x,p) =
\frac{1}{2} \left(\delta f(t,x,p) \pm \delta f(t,x,-p)\right).
\]
Since $U_{\delta f(t)} = U_{\delta f_+(t)}$, 
\begin{align*}
  \partial_t \delta f_- + \T \delta f_+
  &= \nabla U_{\delta f_+} \cdot p\, \phi'(E), \\
  \partial_t \delta f_+ + \T \delta f_- &=0.
\end{align*}
We differentiate the first equation with respect to $t$ and
substitute the second one in order to eliminate $\delta f_+$.
If we write $g$ instead of $\delta f_-$ the linearized (VP) system
takes the form
\begin{equation}\label{linvp_2ndo}
  \partial_t^2 g + \L g = 0,
\end{equation}
where the {\em Antonov operator} $\L$ is defined as
\[ 
\L g \coloneqq - \T^2 g - \Ro g = - \T^2 g
+ \nabla U_{\div j_{g}}\cdot p\, \phi'(E)
\]
with $j_g \coloneqq \int p g dp$; notice that
$\partial_t U_{\delta f_+} = U_{\partial_t\delta f_+} = - U_{\T \delta f_-}$ and
$\rho_{\T \delta f_-}=\div j_{\delta f_-}$. The operator $\Ro$ is the
{\em gravitational response operator}.
We will not go into the functional analysis details of properly defining
these operators on suitable Hilbert spaces. This has been done in
\cite{HaStrRe2022}, in particular, $\L$ can be realized as a self-adjoint
operator on some Hilbert space; since the latter is a weighted
$L^2$ space on $\{f_0>0\}$ with weight $|\phi'(E)|^{-1}$ the key assumption
on the steady state is again that \eqref{mainstabcond} holds.

One can now check that a state $f=f(x,p)$
is an eigenfunction of the operator in \eqref{vplingen} with eigenvalue
$\lambda$ iff $g=f_-$ is an eigenfunction of $\L$ with eigenvalue
$\mu = \lambda^2$. Since the spectrum of $\L$ is real,
the eigenvalues of \eqref{vplingen} come in pairs of the form
$\pm \lambda$ and $\pm i \lambda$ with $\lambda\in \R$. But this means
that the best possible situation as to stability is
that the spectrum of \eqref{vplingen} sits on the imaginary axis, which
is the situation when even in finite dimensions no conclusion to
nonlinear stability is possible; the example in Section~\ref{theexample}
will show that in such a situation stability (in the Lyapunov sense)
cannot even be concluded for the linearized system.

Given the fact that the variational methods by-pass all spectral
considerations and yield nonlinear stability directly, linearization
seemed, for the author, of little value for the questions at hand. But this
conclusion turned out to be too rash for two reasons.
Firstly, variational methods so far do not seem to succeed for (EV),
while linearization has lead to some interesting, non-trivial results.
Secondly, once a steady state is known to be nonlinearly stable
the question arises how exactly it responds to perturbations:
Does it start to oscillate in a time-periodic way or are such oscillations
damped? Very recently, progress on this question was made via linearization
\cite{HaReSchrStr2023,HaStrRe2022,Kunze},
and hence we now take a closer look at \eqref{linvp_2ndo}.

To do so we restrict ourselves to spherically symmetric functions;
$f=f(x,p)$ is spherically symmetric in the sense of \eqref{sphsydef} iff
be abuse of notation
\[ 
f(x,p) = f(r,w,L),\
\mbox{where}\ r=|x|,\ w=\frac{x\cdot p}{r},\ L=|x\times p|^2 .
\]
In order to understand the transport operator $\T$ we must understand
the characteristic flow in the stationary potential $U_0=U_0(r)$.
\begin{lemma} \label{sphsymmchar}
  \begin{itemize}
  \item[(a)]
    Under the assumption of spherical symmetry the characteristic system
    \eqref{chars_vp} for the stationary potential $U_0$ takes the form
    \begin{equation}\label{chars_rwL}
    \dot r = w,\ \dot w = - \Psi'_L(r),\ \dot L =0, 
    \end{equation}
    where the effective potential $\Psi_L$ is defined as
    \[ 
    \Psi_L \colon \left] 0, \infty \right[ \to \R ,\
    \Psi_L(r) \coloneqq U_0(r) + \frac L{2 r^2}.
    \]
    The particle energy $E$ is conserved and takes the form
    \[
    E=E(r,w,L)=\frac{1}{2} w^2 + \Psi_L(r).
    \]
  \item[(b)]
    For any $L>0$ there exists a unique ${r_L} > 0$ such that 
     $\min \Psi_L = \Psi_L(r_L) < 0$,
    and for any $E \in ] \Psi_L (r_L) , 0 [$
    there exist two unique radii ${r_\pm (E,L)}$ satisfying
    \begin{align*}
      0 < r_-(E,L) < r_L < r_+(E,L) < \infty\ \mbox{and}\
      \Psi_L (r_\pm(E,L)) = E.
    \end{align*}
  \item[(c)]
    Let $t \mapsto (r(t), w(t), L)$ be a solution of
    \eqref{chars_rwL}
    with $\Psi_L(r_L)< E = E(r(t), w(t), L)<0$.
    Then $r(t)$ oscillates between $r_-(E,L)$ and $r_+(E,L)$,
    and the period of this motion, i.e., the time needed for $r(t)$ to
    travel from $r_-(E,L)$ to $r_+(E,L)$ and back, is given
    by the {\em period function} of the steady state,
    \[ 
    T(E,L) \coloneqq 2 \int_{r_-(E,L)}^{r_+(E,L)} \frac{dr}{\sqrt{2E - 2\Psi_L(r)}}.
    \]
  \end{itemize}
\end{lemma}
These assertions a fairly easy to see. The key property of the effective
potential is that it has a single well structure. Since
\[
\Psi_L'(r) = \frac{1}{r^3} \left(r m(r) - L\right)
\]
and $r\mapsto r m(r)$ is strictly increasing from $0$ to $\infty$,
$\Psi_L'$ has a unique zero $r_L>0$, and $\Psi_L$ is strictly
decreasing on $]0,r_L]$ with $\lim_{r\to 0}\Psi_L(r) = \infty$,
and strictly increasing on $[r_L,\infty[$ with $\lim_{r\to \infty}\Psi_L(r) = 0$;
the mass function $m(r)$ is defined as in \eqref{mdef}.

The structure of the stationary characteristic flow
can be used to introduce action-angle variables
on the set
\begin{align*}
  D \coloneqq \{(r,w,L)\in\R^3 \mid f_0(r,w,L) > 0\};
\end{align*}
this is a slight abuse of the notation introduced in Prop.~\ref{ssfamilies}.
For $(r,w,L)\in D$ let $(R,W)(\cdot,r,w,L)$
be the solution to \eqref{chars_rwL} with
$(R,W)(0,r,w,L)=(r,w)$; $(R,W)(\cdot,r,w,L)$
is periodic with period $T(E,L)$, where $E=E(r,w,L)$.
We supplement the action variables $(E,L)$ with the
angle variable $\theta\in[0,1]$ defined by
\[
(r,w,L) = \left( (R,W)(\theta T(E,L), r_-(E,L),0,L), L \right);
\]
the mapping
$[0,\frac{1}{2}]\ni\theta\mapsto R(\theta T(E,L),r_-(E,L),0,L)\in
[r_-(E,L),r_+(E,L)]$
is bijective with inverse
\[
\theta(r,E,L) \coloneqq
\frac{1}{T(E,L)} \int_{r_-(E,L)}^r \frac{ds}{\sqrt{2E-2\Psi_L(s)}}.
\]
Functions defined on $D$ can now be written as functions
of the action-angle variables $(E,L,\theta)$.
By the chain rule,
\begin{align}
  \left(\T^2 g\right) (E,L,\theta)
  &=
  \frac{1}{T^2 (E,L)} (\partial_\theta^2 g) (E,L,\theta)\label{Tacang}
\end{align}
for suitable functions $g$ defined on $D$.

One can now analyze the spectra of  $-\T^2$ and $\L$.
Using \eqref{Tacang} and the fact that $\Ro$ is relatively $\T^2$-compact,
one can show that the essential spectra of $\L$ and $-\T^2$ coincide,
and
\begin{align}\label{essspectr_vp}
  \sigma_{ess}(\L)=\sigma_{ess}(-\T^2)=\sigma (-\T^2)
  = \overline{ \left\{  \frac{4\pi^2k^2}{T^2(E,L)}  ~\Big|~ k\in\N_0,\
  (E,L)\in\mathring D^{EL} \right\} },
\end{align}
where $D^{EL}=(E,L)(D)$, cf.~\cite[Theorems~5,7, 5.9]{HaStrRe2022}.
For suitable steady states the period function $T(E,L)$
is bounded from above and bounded away from zero on the support
of the steady state.
Thus \eqref{essspectr_vp} shows that the essential spectrum has
a gap between $0$ and the value $\frac{4\pi^2}{\sup^2(T)}$,
the {\em principal gap} $G$. 

The spectrum of $-\T^2$ is purely essential, but the spectrum of $\L$
may contain isolated eigenvalues, in particular, eigenvalues in
the principal gap $G$. To obtain such eigenvalues a version
of the  Birman-Schwinger principle has been developed, inspired
by a paper by {\sc Mathur}~\cite{Ma}.
It is easily checked that $\lambda\in G$
is an eigenvalue of $\L$ iff $1$ is an eigenvalue of the operator 
\[
Q_\lambda = \Ro\, \left(-\T^2 - \lambda \right)^{-1}.
\]
The operator $Q_\lambda$ is not easy to analyze directly,
but due to spherical symmetry,
\[ 
(\Ro g)(r,w,L) = -\frac{4\pi^2}{r^2}\,w\,\phi'(E)
\int_{-\infty}^\infty \int_0^\infty \tilde w\, g(r,\tilde w,\tilde L)
\,d\tilde L\, d\tilde w.
\]
Hence $\Ro$ and $Q_\lambda$ map onto
functions of the form $\vert\phi'(E)\vert\,w\,F(r)$ which
allows the definition of an operator
\[
\M_\lambda \colon \F \to \F 
\]
on a Hilbert space of functions of the radial variable $r$ such that
any eigenvalue of $\M_\lambda$ gives an eigenvalue of $Q_\lambda$.
When considered on the appropriate function space $\F$ this
{\em Mathur operator} is a symmetric Hilbert-Schmidt operator with
an integral kernel representation. The largest element in its spectrum,
which is an eigenvalue,
is given by
\[
M_\lambda = 
\sup\left\{ \langle h,\M_\lambda h\rangle_{\F} \mid
h\in\F,\,\|h\|_{\F}=1\right\}.
\]
It follows that the operator $\L$ has an eigenvalue
in the principal gap $G$ iff there exists $\lambda\in G$
such that $M_\lambda\geq 1$, cf.~\cite[Theorem~8.11]{HaStrRe2022}.
This criterion can be verified for certain examples of steady states
by rigorous proof, and for more general examples with numerical support,
cf.~\cite[Section~8.2]{HaStrRe2022}.

A positive eigenvalue of $\L$ gives rise to a time-periodic, oscillating
solution of the linearized (VP) system \eqref{vplingen}, and this
explains---at least on the linear level---numerical observations made
in \cite{RaRe2018}; the fact that the latter
oscillating solutions pulse in the sense that their support expands
and contracts can be understood by linearization in mass-Lagrange variables,
which leads to the same spectral problem,
cf.~\cite[Section~3.2]{HaStrRe2022}.

But in  \cite{RaRe2018} it was also observed that some steady states
upon perturbation start to oscillate in a damped way.
In \cite{HaReSchrStr2023} such damping phenomena are for the
first time rigorously analyzed in the gravitational situation.
A family of steady states of (VP) with a point mass at the center
is constructed, which are parameterized by their polytropic index $k>1/2$, 
so that the phase space density of the steady state is $C^1$ at the
vacuum boundary if and only if $k>1$; see Remark~(e) at the end
of Section~\ref{stst}. 
The following dichotomy result is established:
If $k>1$, linear perturbations damp, and if $1/2< k\le1$ they do not.
The undamped oscillations for $1/2< k\le 1$ are obtained by Birman-Schwinger
type arguments as above.
The damping for $k>1$ occurs on the level of macroscopic quantities and is
(up to now)
non-quantitative: No damping rate is established, but (for example)
\[ 
\lim_{T\to\infty}\frac{1}{T}\int_0^T\|\nabla U_{\T f(t,\cdot)}\|_{L^2}^2 dt = 0,
\]
where
$[0,\infty[\ni t\to f(t)$ is any solution to~\eqref{linvp_2ndo}
with initial data $f(0)$ in the domain of $\L$.

This type of damping is obtained by an application of the
RAGE theorem~\cite{CyFrKiSi}. The main fact which has to be
established in order to apply this theorem is that the operator
$\L$ has no eigenvalues, and the key difficulty is to exclude
eigenvalues embedded in the essential spectrum,
cf.~\cite[Theorem~4.5]{HaReSchrStr2023}.

The damping result can also be viewed as a result on macroscopic, asymptotic
stability for the corresponding steady states on the linearized level.
The importance of relaxation processes in astrophysics can be seen from
the discussion in \cite{BiTr} and the references there; we explicitly mention
the pioneering work of {\sc Lynden-Bell}~\cite{LB1962,LB1967}.

In the plasma physics situation an analogous damping phenomenon
around spatially homogeneous steady states was discovered by
{\sc Landau}~\cite{Landau1946} on the linearized level, and
on the nonlinear level in the celebrated work of
{\sc Mouhot} and {\sc Villani}~\cite{MoVi2011},
see also~\cite{BeMaMo2016,GrNgRo2020}.
It should be noticed that in this case the characteristic flow of the
unperturbed steady state is simple free streaming,
so the corresponding result for the gravitational case
has to deal with substantial and qualitatively new difficulties
due to the non-trivial characteristic steady state flow. 
\section{Stability for (EV)---steady states with small central redshift}
\label{stabkappasmall}
\setcounter{equation}{0}
\subsection{The set-up}
\label{evstab_setup}
We consider the spherically symmetric (EV) system as formulated in
Section~\ref{ss_ev}and choose the formulation which employs the
non-canonical momentum variable $v$, cf.~\eqref{vdef}.
Functions or states $f=f(x,v)\geq 0$ are always spherically symmetric, i.e.,
\begin{equation} \label{rwL_ev}
  f(x,v)=f(r,w,L)\ \mbox{with}\ r=|x|,\ w=\frac{x\cdot v}{r},\
  L=|x\times v|^2,
\end{equation}
and induce the mass-energy density
\[
\rho_f (r)=\rho_f (x) = \int \pv f(x,v)\,dv
\]
and the metric component $\lambda=\lambda_f$ via
\begin{equation} \label{lambdadef}
  e^{-2\lambda_f(r)} = 1-\frac{2 m_f(r)}{r}
  = 1 - \frac{8\pi}{r} \int_0^r \rho_f(s)\, s^2 ds;
\end{equation}
only states $f$ with $2 m_f (r) < r$ are admissible.

Let us fix some steady state $(f_0,\lambda_0,\mu_0)$ of (EV) of the form
\eqref{ststansatz_ev} with \eqref{phicond0}; for the moment the central
redshift $\kappa$ of this steady state is not relevant and suppressed.
Let us also fix a function $\Phi\in C^1([0,\infty[)$ with $\Phi(0)=0$.
In Section~\ref{ss_ev} we introduced the energy and Casimir functionals,    
see \eqref{energy_evv}, \eqref{casi_evv}, and we define
the energy-Casimir functional
\[
\H_C (f)
\coloneqq
\H (f) + \Cc (f)
=
\iint \pv f(x,v)\,dv\,dx + \iint e^{\lambda_f} \Phi(f(x,v))\,dv\,dx.
\]
We formally expand $\H_C$ about $f_0$:
\[ 
\mathcal{H}_C (f_0+\delta f) = 
\mathcal{H}_C (f_0)+D\mathcal{H}_C(f_0)(\delta f)
+ D^2\mathcal{H}_C(f_0)(\delta f,\delta f) + \mathrm{O}((\delta f)^3).
\]
To proceed we again make the standard stability assumption
that on the support of the steady state $\phi$ is strictly decreasing,
cf.~\eqref{mainstabcond}.
If the function $\Phi$ is such that
\[
\Phi'(f_0) = \Phi' (\phi(E)) = - E,\ \mbox{i.e.}\ \Phi' = - \phi^{-1},
\]
then a non-trivial, formal computation \cite{HaRe2013},
see also \cite{KM}, shows that
\[
D\mathcal{H}_C(f_0)(\delta f) =0
\]
and
\begin{equation}\label{qecdef}
  D^2\mathcal{H}_C(f_0)(\delta f,\delta f)
  =
  \frac{1}{2}\iint \frac{e^{\lambda_0}}{|\phi'(E)|}(\delta f)^2\,dv\,dx
  - \frac{1}{2}\int_0^\infty e^{\mu_0 - \lambda_0}
  \left(2 r \mu_0' +1\right)\, (\delta\lambda)^2\,dr.
\end{equation}
Here $\delta\lambda$ is to be expressed in terms of $\delta f$ through the
variation of (\ref{lambdadef}), cf.\ (\ref{varlambdadef}) below.
Only perturbations $\delta f$ which are supported in the support
of the steady state $f_0$ are considered---$\delta f$ must
be small compared to $f_0$---which is important for the
first integral in~\eqref{qecdef}.
We see that the steady state is a critical point of the
energy-Casimir functional $\H_C$, but like for (VP) the quadratic term
\eqref{qecdef} is the sum of two terms with opposite signs,
which is the central difficulty in the stability analysis;
one should notice that since $\mu_0' \geq 0$,
\begin{equation}\label{muprtermpos}
2 r \mu_0' +1 \geq 1,\ r\geq 0.
\end{equation}
For (VP), one way to by-pass this difficulty was the global
minimizer approach explained in Section~\ref{globvar_vp},
but so far this strategy has not been successful for (EV)
for reasons which we indicated in Section~\ref{ss_criticality}.
But we also saw in Section~\ref{locvar_vp} how for (RVP)
$D^2 {\cal H}_c(f_0)$ is
positive definite on linearly dynamically
accessible states, and how this fact can lead to a stability result
as well. We follow this route in the present (EV) case.
To do so we first need to discuss the concept of dynamically accessible
states for (EV).

An admissible state $f$ is
{\em nonlinearly dynamically accessible from $f_0$} iff
for all $\chi \in C^1(\R)$ with $\chi(0)=0$,
\begin{equation}\label{nldynac}
  \Cc_\chi (f)=\Cc_\chi (f_0),
\end{equation}
where $\Cc_\chi$ is defined like $\Cc$, but with the general function
$\chi$ instead of $\Phi$, the latter being specific for the steady state
under consideration. Property \eqref{nldynac} is preserved by the flow of
the Einstein-Vlasov system. 
Taking the first variation in (\ref{nldynac}), a
definition for $\delta f$ to be linearly dynamically accessible could
be that
\begin{equation}\label{ldynac}
  D \Cc_\chi(f_0)(\delta f)
  =\iint e^{\lambda_0} \left(\chi'(f_0)\delta f
  +\chi(f_0)\delta\lambda\right) dv\,dx = 0
\end{equation}
for all $\chi \in C^1(\R)$ with $\chi(0)=0$, where
\begin{equation} \label{varlambdadef}
  \delta\lambda = 
  e^{2 \lambda_0}\frac{4 \pi}{r} \int_0^r s^2 \rho_{\delta f} (s)\,ds.
\end{equation}
This needs to be turned into a more explicit and workable definition.
A suitable integration by parts turns \eqref{ldynac} into
\begin{equation} \label{ldynac2}
  D \Cc_\chi(f_0)(\delta f)
  =\iint e^{\lambda_0}\chi'(f_0)\left[\delta f  
  - e^{\mu_0}\delta\lambda \, \phi'(E)\frac{w^2}{\pv }\right] dv\,dx = 0,
\end{equation}
cf.~\cite[Lemma~3.1]{HaRe2013}. Hence a variation $\delta f$ satisfies
\eqref{ldynac}, if
\begin{equation} \label{ldynacexpl}
  e^{\lambda_0} \delta f -e^{\mu_0+\lambda_0}\delta\lambda\, 
  \phi'(E)\frac{w^2}{\pv }=\{h,f_0\}
\end{equation}
for some spherically symmetric generating function $h\in C^2(\R^6)$;
note that for any such $h$,
\[
\iint\chi'(f_0)\,\{h,f_0\}\,dv\,dx =0.
\]
We make the definition more explicit; recall that $D=\{f_0>0\}$.
\begin{definition}
  A state $\delta f$ is
  {\em linearly dynamically accessible from $f_0$} if there
  exists some spherically symmetric generating function $h\in C^1(\overline{D})$
  such that
  \begin{equation} \label{ldynacdef}
    \delta f = f_h \coloneqq
    e^{-\lambda_0}\{h,f_0\}
    +4\pi r e^{2\mu_0+\lambda_0}\phi'(E) \frac{w^2}{\pv}
    \int \phi'(E(x,\tilde v))\,h(x,\tilde v)\,\tilde w\,d\tilde v.
  \end{equation}
\end{definition}
Notice that possible values of the generating function $h$ outside $D$
would not influence $\delta f$ which vanishes outside $D$.
The justification for this definition is the following result,
cf.~\cite[Prop.~3.2]{HaRe2013}; we will see later that this form
of $\delta f$ is preserved under the linearized (EV) dynamics,
and we will give a slightly more general, functional-analysis type
definition of this concept.
\begin{proposition} \label{ldynacprop}
  If $\delta f$ is linearly dynamically accessible from $f_0$
  and $\delta\lambda$ is defined by
  (\ref{varlambdadef}), then
  \begin{equation}\label{ldynaclambda}
    \delta\lambda = \lambda_h \coloneqq 4\pi r e^{\mu_0+\lambda_0}
    \int \phi'(E)\,h(x,v)\,w\,dv,
  \end{equation}
  $\delta f$ satisfies both (\ref{ldynacexpl}) and (\ref{ldynac}), and
  \begin{equation} \label{dynacdef}
    \delta f = f_h =
    \phi'(E)\,\left(e^{-\lambda_0}\{h,E\}
    + e^{\mu_0} \lambda_h \,\frac{w^2}{\pv }\right).
  \end{equation}
\end{proposition}
The key feature of linearly dynamically accessible states is that
if we substitute such a state into $D^2\mathcal{H}_C(f_0)$, then,
for sufficiently non-relativistic steady states, this quadratic
form becomes positive definite, just as for (RVP), cf.~Lemma~\ref{kandrup_rvp}.
To see this we have to understand the behavior of the steady states
obtained in Prop.~\ref{ssfamilies} for small redshift $\kappa$.
\subsection{Steady states for $\kappa$ small---the non-relativistic limit}
\label{stst_nonrel}
We fix an ansatz function $\varphi$ satisfying \eqref{phicond0}, define
\[
\varphi_N(\eta) \coloneqq C \eta^k\ \mbox{for}\ \eta >0
\]
with $0<k<3/2$ and $C>0$, and require that
\begin{equation}\label{phicond2}
  \varphi(\eta) = \varphi_N(\eta) + \mathrm{O}(\eta^{k+\delta})
  \ \mbox{for}\ \eta \to 0+,
\end{equation}
with some $\delta >0$; notice that this condition implies \eqref{phicond1}.
For $\kappa >0$ small we wish to relate 
$y_\kappa$ and the induced steady state
$(f_\kappa,\lambda_\kappa,\mu_\kappa)$ obtained in Prop.~\ref{ssfamilies}~(b)
to the solution $y_N$ of the Newtonian problem \eqref{yeqvp}, with
$y_N(0)=1$ and $\varphi_N$ as Newtonian microscopic
equation of state, and the induced
steady state $(f_N,U_N)$ of (VP).
We define
\[
a \coloneqq \frac{k+1/2}{2}.
\]
\begin{proposition}\label{kappato0}
  There exist constants $\kappa_0 >0$, $S_0>0$, and $C>0$ such that for all
  $\kappa \in ]0,\kappa_0]$,
  \[ 
  \supp \rho_\kappa \subset [0,\kappa^{-a}S_0],
  \]
  and for all $r\geq 0$,
  \[ 
  \left|\kappa^{-1} y_\kappa (r) -y_N (\kappa^a r)\right| \leq C \kappa^\delta,
  \]
  \[ 
  \left|e^{2\lambda_\kappa(r)} -1\right|\leq C \kappa,
  \]
  \[ 
  \left|\kappa^{-1} \mu_\kappa (r) - U_N (\kappa^a r)\right|+
  |\kappa^{-1-2a}\rho_\kappa (r) - \rho_N(\kappa^a r)| \leq C \kappa^\delta.
  \]
\end{proposition}
This result was shown in \cite{HaRe2014}. For the proof one
introduces a rescaled function $\bar{y}_\kappa$ 
and a rescaled radial variable $s$ by
\[ 
y_\kappa (r) = \kappa \, \bar{y}_\kappa (\kappa^a r) 
= \kappa \, \bar{y}_\kappa (s),\quad
s = \kappa^a r.
\]
One can then derive an equation for the function $\bar{y}_\kappa$
which corresponds to the equation~\eqref{yeq} for $y_\kappa$.
In this rescaled version of \eqref{yeq} the microscopic equation of state
becomes
\[
\varphi_{\kappa}(\eta) \coloneqq
\kappa^{-k}\varphi(\kappa \eta),
\]
which by \eqref{phicond2} converges to $\varphi_N$ for $\kappa \to 0$.
In addition, the ``relativistic corrections'' in the rescaled version
of \eqref{yeq} like the pressure term $\sigma$ and the term $2 m/r$
in the denominator pick up multiplicative factors of $\kappa$,
while  $\bar{y}_\kappa(0)=1=y_N(0)$. A lengthy Gronwall-type argument
implies that there exist constants $\kappa_0>0$ and $C>0$
such that for all $0<\kappa\leq \kappa_0$ and $s\geq 0$,
\[
|\bar{y}_\kappa (s) -y_N (s)| \leq C \kappa^\delta.
\]
The assertions in Prop.~\ref{kappato0} then follow.

In Section~\ref{ss_linvp} we saw that action-angle variables
are an essential tool for understanding the linearized dynamics
in the (VP) case. Introducing these variables relied on the single-well
structure of the effective potential $\Psi_L$ discussed in
Lemma~\ref{sphsymmchar}.
For (EV), the steady state characteristics obey the equations
\[
\dot r =  e^{-\lambda_\kappa (r)}\,\partial_w E_\kappa (r,w,L),\ 
\dot w = -e^{-\lambda_\kappa (r)}\,\partial_r E_\kappa (r,w,L)
\]
with
\[
E_\kappa (r,w,L) = e^{\mu_\kappa(r)}\sqrt{1+w^2 + \frac{L}{r^2}}.
\]
Let us define the analogue of the Newtonian effective potential as
\[
\Psi_{\kappa,L}(r) \coloneqq e^{\mu_\kappa(r)}\sqrt{1 + \frac{L}{r^2}}
\]
and assume that
\begin{equation} \label{singwell_ev}
  \frac{2 m_\kappa (r)}{r} \leq \frac{1}{3},\ r>0.
\end{equation}
Then one can show that $\Psi_{\kappa,L}$ has a single-well structure
analogous to Lemma~\ref{sphsymmchar}~(b), and the conclusions of
Lemma~\ref{sphsymmchar}~(c) and its action-angle consequences
remain valid, cf.~\cite[Section~3]{GueStrRe}. By Prop.~\ref{kappato0},
the condition \eqref{singwell_ev} holds for $\kappa$ small.
\begin{remark}
  The question whether the steady state characteristic flow has the
  single-well structure is intimately related to the question whether
  for spherically symmetric steady states the phase space density
  can always be written in the form $f=\phi(E,L)$. For (VP) this
  result, which is sometimes called Jeans' theorem, is a direct
  consequence of the single-well structure of the effective potential.
  For (EV), Jeans' theorem is known to be false, cf.~\cite{Sch2}.
  Numerical evidence strongly suggests that for isotropic steady
  states of (EV), $2 m (r)/r < 1/2$, which is a considerably
  sharper bound than the general Buchdahl inequality
  \cite{Andr2007,Andr2008,Buch},
  but it is unclear
  whether this is sufficient to yield the single-well structure.
\end{remark}
The information provided by Prop.~\ref{kappato0} can be used to show
that on linearly dynamically accessible states the quadratic
form $D^2\mathcal{H}_C(f_\kappa)$ is positive definite for $\kappa$ sufficiently
small.
\subsection{An energy-Casimir coercivity estimate}
\label{coerc_ec_ev}
As in the previous section, let the microscopic equation of state
$\varphi$ satisfy \eqref{phicond0} and \eqref{phicond2}, and let
$(f_0,\lambda_0,\mu_0)$ be a steady state
as obtained in Prop.~\ref{ssfamilies}~(b).
As an abbreviation, let
\begin{align*}
  \mathcal{A}(\delta f,\delta f)
  \coloneqq&
  D^2\mathcal{H}_C(f_0)(\delta f,\delta f) \\
  =&
  \frac{1}{2}\iint \frac{e^{\lambda_0}}{|\phi'(E)|}(\delta f)^2\,dv\,dx
  - \frac{1}{2}\int_0^\infty e^{\mu_0 - \lambda_0}
  \left(2 r \mu_0' +1\right)\, (\delta\lambda)^2\,dr.
\end{align*}
The following result is the desired
energy-Casimir coercivity estimate.
\begin{theorem}\label{th:coercivity}
  There exist constants $C^\ast >0$ and $\kappa^\ast>0$ such that
  for any $0<\kappa\leq\kappa^\ast$ and any
  spherically symmetric function $h\in C^1(\overline{D})$ 
  which is odd in $v$,
  \[
  \mathcal{A}(\delta f,\delta f)\geq 
  C^\ast\iint|\phi'(E)|\,\left((rw)^2 
  \left|\left\{E,\frac{h}{rw}\right\}\right|^2 
  + \kappa^{1+2a} |h|^2\right) dv\,dx,
  \]
  where $h$ generates the dynamically accessible perturbation
  $\delta f$ according to (\ref{ldynacdef}).
\end{theorem}
For a dynamically accessible 
perturbation $\delta f=f_h$ defined by~\eqref{ldynacdef}
and~\eqref{ldynaclambda}, 
\[
\A(\delta f,\delta f)=\A(h,h)=
\frac{1}{2} \A_{1}(h)+\frac{1}{2}\A_{2}(h),
\]
where 
\begin{align*}
  \A_{1}(h)
  \coloneqq&
  \iint e^{-\lambda_0}|\phi'(E)||\{E,h\}|^2 dv\,dx
  -\int_0^{\infty}e^{\mu_0-\lambda_0}
  (2r{\mu_0}'+1)(\lambda_h)^2 dr,\\
  \A_{2}(h)
  \coloneqq&
  -2 \iint |\phi'(E)|\{E,h\} 
  \lambda_h e^{\mu_0}\frac{w^2}{\pv}\,dv\,dx
  + \iint|\phi'(E)|e^{2\mu_0+\lambda_0}\frac{w^4}{\pv^2}(\lambda_h)^2
  dv\,dx .
\end{align*}
It turns out that $\A_{1}$ yields the desired
lower bound while $\A_{2}$ is of higher order in $\kappa$ and can
be controlled by the positive contribution from $\A_{1}$.
For more details we refer to \cite[Theorem~5.1]{HaRe2014}, but we
emphasize that for the proof the complete structure of the
stationary (EV) system must be exploited, in particular, the
static version of \eqref{ein4} and \eqref{laplusmu} come into play.

The assumption in Theorem~\ref{th:coercivity} that $h$ is odd in $v$
can be removed.
We split a general, spherically symmetric function $h\in C^2(\R^6)$
into its even and odd-in-$v$ parts, $h=h_+ + h_-$. Then 
\[
\lambda_h=4\pi r e^{\mu_0+\lambda_0}
\int \phi'(E)\,h_-(x,v)\,w\,dv = \lambda_{h_-},
\]
and
\[
\delta f_+ = (f_h)_- =
e^{-\lambda_0}\{h_+,f_0\},\
\delta f_- = (f_h)_+ =
e^{-\lambda_0}\{h_-,f_0\} + 
\gamma e^{\mu_0} \phi'(E) \frac{w^2}{\pv} \lambda_{h_-}.
\]
Hence
\begin{align}\label{gen_coercivity}
  \A (\delta f,\delta f)
  &=
  \A (\delta f_+,\delta f_+) + 
  \iint e^{\lambda_0}\frac{\delta f_+ \, \delta f_-}{|\phi'(E)|}\,dv\,dx
  + \frac{1}{2}
  \iint e^{\lambda_0}\frac{|\delta f_-|^2}{|\phi'(E)|}\,dv\,dx\notag\\
  &=
  \A (\delta f_+,\delta f_+) +
  \frac{1}{2}
  \iint e^{-\lambda_0} |\phi'(E)|\, \left|\{E,h_+\}\right|^2 dv\,dx\notag\\
  &\geq
  C^\ast\iint|\phi'(E)|\,\left((rw)^2 
  \left|\left\{E,\frac{h_-}{rw}\right\}\right|^2 
  + \kappa^{1+2a} |h_-|^2\right) dv\,dx\notag\\
  &\quad {}+
  \frac{1}{2}
  \iint e^{-\lambda_0} |\phi'(E)|\, \left|\{E,h_+\}\right|^2dv\,dx.
\end{align}

If one tries to proceed from this positive definiteness result on
the second variation of $\H_C$ towards nonlinear stability,
analogously to Section~\ref{locvar_vp} for (RVP), serious difficulties
in deriving an analogue of Theorem~\ref{locmin_rvp} arise,
again from the inherent lack of compactness for (EV).
We believe that some analogue of Theorem~\ref{locmin_rvp} remains correct for
(EV), but so far the result of the present section has only been used to
derive linear stability.
\subsection{Linear stability}
\label{linstab_ev}
In order to deal with this issue we need to linearize (EV) about
some given steady state
$(f_0,\lambda_0,\mu_0)$; for the moment, no assumption is made on
the size of $\kappa$, because we will use the linearized system
also for large $\kappa$. We substitute 
\[
f(t)=f_0 +\delta f(t),\ 
\lambda(t)=\lambda_0 + \delta\lambda(t),\ 
\mu(t)=\mu_0 + \delta\mu(t)
\]
into the system, use the fact that $(f_0,\lambda_0,\mu_0)$
is a solution, and drop all terms beyond the linear ones in
$(\delta f,\delta\lambda,\delta\mu)$. In addition the boundary conditions
$\delta\lambda(t,0)=\delta\lambda(t,\infty) =\delta\mu(t,\infty)=0$
are imposed.
We observe that
\begin{equation}\label{linein1}
  \delta\lambda = \lambda_{\delta f} \coloneqq 
  e^{2 \lambda_0}\frac{4 \pi}{r} \int_0^r s^2 \rho_{\delta f} (s)\,ds
\end{equation}
is the corresponding solution to the linearized version of the
field equation \eqref{ein1},
cf.~\eqref{varlambdadef}. The linearized versions of the
field equations \eqref{ein2}, \eqref{ein3} yield
\begin{align}
  \delta\mu'
  &= \mu'_{\delta f} \coloneqq
  4\pi r e^{2\lambda_0} \sigma_{\delta f}
  + \left(2\mu_0'+\frac{1}{r}\right) \lambda_{\delta f}, \label{linein2}\\
  \dot{\delta\lambda}
  &=
  -4\pi r e^{\mu_0+\lambda_0}\jmath_{\delta f},\label{linein3}
\end{align}
where as before,
\[ 
\rho_{\delta f} = \int \pv \delta f\, dv,\
\sigma_{\delta f} = \int \frac{w^2}{\pv} \delta f\, dv,\
\jmath_{\delta f} = \int w\, \delta f\, dv.
\]
If we substitute into the linearization of the Vlasov equation
\eqref{vlasov_evv},
\begin{align}\label{linvlasov}
  \partial_t\delta f
  &+
  e^{-\lambda_0}\{\delta f,E\} \notag \\
  &{}+
  4\pi r e^{2\mu_0+\lambda_0} \phi'(E)
  \left(\frac{w^2}{\pv} \jmath_{\delta f} - w \sigma_{\delta f}\right)
  - e^{2\mu_0-\lambda_0} \left(2\mu_0'+\frac{1}{r}\right)
  \lambda_{\delta f}\phi'(E)\, w =0;
\end{align}
it can be shown that \eqref{linein3} follows from the other equations.

In order to proceed we need to observe that linear dynamic accessibility
propagates under the dynamics of the linearized (EV) system.
We recall that some $h\in C^1(\overline{D})$ generates a
linearly dynamically accessible perturbation $\delta f = f_h$ according
to \eqref{ldynacdef}, and \eqref{linein1} turns into
$\delta \lambda = \lambda_h$, defined by \eqref{ldynaclambda}.
One can check that if $t\mapsto h(t)$ solves the
transport equation
\begin{equation}\label{heqn}
  \partial_t h + e^{-\lambda_0}\{h,E\} +
  e^{\mu_0} \lambda_{h(t)}  \frac{w^2}{\pv } -
  e^{\mu_0} \pv \mu_{h(t)} =0
\end{equation}
with spherically symmetric initial data $h(0)=\mathring{h}$,
then $\delta f (t) = f_{h(t)}$ defined according to (\ref{ldynacdef})
is the solution of the above
linearized (EV) system
to the linearly dynamically accessible data
$\mathring{\delta f}=f_{\mathring{h}}$. In particular,
$\delta f(t)$ is the linearly dynamically accessible state generated
by $h(t)$.
A simple iteration argument shows that for any
$\mathring{h} \in C^1(\overline{D})$
there exists a unique solution
$h\in C^1([0,\infty[;C(\overline{D}))
 \cap C ([0,\infty[;C^1(\overline{D}))$
of \eqref{heqn} with $h(0)=\mathring{h}$.
The induced linearly dynamically accessible solution
$\delta f$ needs to be only continuous
(unless one demands more regularity of $\mathring{h}$ and $\varphi$)
and solves the linearized (EV) system integrated along the steady
state characteristics, cf.~\cite{BaMoRe95} for the analogous concept for
the linearized (VP) system. We do not discuss these issues further, since in
the context of the more functional analytic approach in Section~\ref{FA}
we solve the linearized (EV) system by a suitable $C_0$ group. 

The important fact here is that
such linearly dynamically accessible solutions 
preserve the energy
\begin{equation}\label{linen}
  \mathcal{A}(f_h,f_h)
  = \mathcal{A}(h,h)
  =
  \frac{1}{2}\iint \frac{e^{\lambda_0}}{|\phi'(E)|}(f_h)^2\,dv\,dx
  - \frac{1}{2}\int_0^\infty e^{\mu_0 - \lambda_0}
  \left(2 r \mu_0' +1\right) (\lambda_h)^2\,dr.
\end{equation}
Combining this fact with Theorem~\ref{th:coercivity} or with the
more general estimate \eqref{gen_coercivity} proves the following
stability result.
\begin{theorem}\label{linstabthm}
  Let $C^\ast$ and $\kappa^\ast$ be as in Theorem~\ref{th:coercivity},
  and let $0<\kappa \leq \kappa^\ast$. Then the steady state
  $(f_\kappa,\lambda_\kappa,\mu_\kappa)$
  is linearly stable in the following sense. 
  For any spherically symmetric function
  $\mathring{h}\in C^1(\overline{D})$ the solution of the linearized
  (EV) system with dynamically accessible data $\mathring{\delta f}$
  generated by $\mathring{h}$
  according to (\ref{ldynacdef}) satisfies
  for all times $t\geq 0$ the estimate
  \begin{align*}
    C^\ast\iint|\phi'(E)|\,\left((rw)^2 
    \left|\left\{E,\frac{h_-(t)}{rw}\right\}\right|^2 
    + \kappa^{1+2a} |h_-(t)|^2\right) dv\,dx 
    &\\
    +
    \frac{1}{2}
    \iint e^{-\lambda_\kappa} |\phi'(E)|\, \left|\{E,h_+(t)\}\right|^2 dv\,dx
    &\leq
    \mathcal{A}(\mathring{\delta f}).
  \end{align*}
\end{theorem}
The restriction to perturbations 
$\mathring{\delta f}$ of the form~\eqref{ldynacdef} may seem a bit special.
Condition~\eqref{ldynac2} suggests that the natural
set of perturbations for the linear problem are functions
$\mathring{\delta f} \in C^1(\R^6)$ supported on the support of
$f_0$ with the property that
\[
e^{\lambda_0}\mathring{\delta f}
- e^{\mu_0+\lambda_0} \mathring{\delta \lambda}
\phi' (E) \frac{w^2}{\pv}
\ \mbox{is $L^2$-orthogonal to any}\ \psi(f_0) \in L^2(\R^6),
\ \psi\in C (\R).
\]
For any such perturbation 
there exists a generating function $h\in C^2(\R^6)$ so that
\[
\{h,f_0\} = e^{\lambda_0}\mathring{\delta f}
- e^{\mu_0+\lambda_0} \mathring{\delta \lambda} \phi' (E)
\frac{w^2}{\pv}, 
\] 
which by Proposition~\ref{ldynacprop} says that $\mathring{\delta f}$ is
linearly dynamically accessible.
The proof is analogous to the proof of the parallel 
fact for (VP) given in \cite[Section~3.2]{GuRe2007} and
relies on the fact that for $\kappa$ small the stationary characteristic
flow (or rather its effective potential) has a single-well structure,
see the end of Section~\ref{stst_nonrel}.
\section{Instability for (EV)---steady states with large central redshift}
\label{stabkappalarge}
\setcounter{equation}{0}
We continue to use the set-up which we discussed in Section~\ref{evstab_setup}.
We saw in Section~\ref{coerc_ec_ev} that the second variation
$D^2\mathcal{H}_C$ is positive definite on linearly dynamically accessible
states, provided the central redshift $\kappa$ of the steady state
$(f_\kappa,\lambda_\kappa,\mu_\kappa)$ in question is small, and we saw in
Section~\ref{linstab_ev} that this fact implies linearized stability
of the corresponding steady state. The key to this was a good understanding
of the behavior of the steady state in the limit $\kappa\to 0$,
the Newtonian limit. It turns out that for $\kappa$ sufficiently large,
such steady states become unstable. The major step towards this result
is that there is a direction in which $D^2\mathcal{H}_C$ becomes negative,
provided $\kappa$ is sufficiently large, and the key to this is a good
understanding of the behavior of the steady state in the limit
$\kappa\to \infty$. This is more challenging and more interesting,
because no $\kappa\to\infty$ limiting system seems to suggest itself for
the role that (VP) plays as the $\kappa\to 0$ limiting system.
But such a system exists.
\subsection{Steady states for $\kappa$ large---the ultrarelativistic limit}
\label{stst_ultrarel}
We again consider steady states of (EV) as obtained in
Prop.~\ref{ssfamilies}. As indicated in part~(b) of the remark following that
proposition
a microscopic equation of state $\varphi$ gives rise to a
macroscopic equation of state which relates pressure and mass-energy density,
more precisely,
\begin{equation} \label{pphi}
  \sigma_\kappa = P(\rho_\kappa),\ \mbox{where}\ P \coloneqq h\circ g^{-1}
\end{equation}
with $g$ and $h$ defined by \eqref{gdef} and \eqref{hdef}. When $\kappa$ is
very large also $y_\kappa(r)$ and $\rho_\kappa(r)$ become very large at
least for $r$ close to $0$.
For $y$ very large,
\begin{equation} \label{gdef_as}
  g(y) = 4 \pi e^{4 y} \int_0^{1-e^{-y}} \varphi(\eta)\, (1-\eta)^2\,
  \left((1-\eta)^2 - e^{-2y}\right)^{1/2} d\eta
  \approx e^{4y} \eqqcolon g^\ast (y),
\end{equation}
and
\begin{equation} \label{hdef_as}
  h(y) =
  \frac{4 \pi}{3} e^{4y} \int_0^{1-e^{-y}} \varphi(\eta)\,
  \left((1-\eta)^2-e^{-2y}\right)^{3/2} d\eta
  \approx \frac{1}{3} e^{4y} \eqqcolon h^\ast (y), 
\end{equation}
where for the sake of notational simplicity we normalize
\[
4 \pi \int_0^{1} \varphi(\eta)\, (1-\eta)^3\, d\eta = 1.
\]
Hence for $\kappa$ very large and close to the center the
equation of state \eqref{pphi} asymptotically turns into 
\begin{equation} \label{radiation}
  \sigma_\kappa = P^\ast (\rho_\kappa) = \frac{1}{3} \rho_\kappa
\end{equation}
which is known in astrophysics and cosmology as the equation of state for
radiation. It can be shown that
\[ 
\left|P(\rho) -\frac{1}{3} \rho\right| \leq C \rho^{1/2},\ \rho\geq 0
\]
for some constant $C>0$, cf.~\cite{HaLinRe}, which is the precise version
of the limiting behavior of the equation of state.

Of course now the question arises how the limiting
equation of state \eqref{radiation} fits into the Vlasov
context, since that equation of state
cannot come from an isotropic
steady state particle distribution of the form \eqref{ststansatz_ev}:
\[
\sigma (r) =
\int  f(x,v)\, \left(\frac{x\cdot v}{r}\right)^2 \frac{dv}{\pv}
=
\frac{1}{3}\int  f(x,v)\, |v|^2 \frac{dv}{\pv}
<
\frac{1}{3}\int  f(x,v)\, (1+|v|^2) \frac{dv}{\pv}
= \frac{1}{3} \rho(r);
\]
massive particles do not behave like radiation.
To obtain a Vlasov-type system which captures the
limiting behavior as $\kappa \to \infty$ we must
pass to a collisionless ensemble of massless particles.
Mathematically, this means that throughout the (EV) system
the term $\pv$ must be replaced by $|v|$. In particular,
\eqref{ststansatz_ev} turns into the ansatz
\[
f(x,v) = \phi(e^{\mu(r)}|v|) = \varphi\left(1-\frac{e^{\mu(r)}|v|}{E_0}\right).
\]
This ansatz satisfies the massless version of\eqref{vlasov_evv},
we get exactly the radiative equation of state 
\eqref{radiation}, and
\[
\rho(r) = \int \phi(e^{\mu(r)}|v|)\, |v|\, dv
= 4\pi \int_0^\infty \phi(\eta)\, \eta^3 d\eta \, e^{-4\mu(r)}
\]
which is as expected from \eqref{gdef_as}.
Hence if $y$ is a solution 
of \eqref{yeq} where $g$ and $h$ are replaced by
$g^\ast$ and $h^\ast$, and
$\mu, \lambda, \rho, \sigma$ are induced by $y$,
then these quantities satisfy
the stationary Einstein equations together with the radiative equation of state
and the above $f$ is a consistent, stationary solution of the
massless Einstein-Vlasov system.

Proceeding as in \cite[Theorem~3.4]{Rein94} one can obtain the following result.


\begin{lemma}
  For every $\kappa >0$ there exists a unique
  solution $y^\ast=y^\ast_\kappa\in C^1([0,\infty[)$ to the problem
  \begin{equation} \label{grzeq}
    y'(r)= - \frac{1}{1-2 m^\ast(r)/r} \left(\frac{m^\ast(r)}{r^2}
    + 4 \pi r \sigma^\ast(r)\right) ,\ 
    y(0)=\kappa>0,
  \end{equation}
  where $\rho^\ast=g^\ast(y)$, $\sigma^\ast=h^\ast(y)$ with
  \eqref{gdef_as}, \eqref{hdef_as} and
  \[ 
  m^\ast(r) = m^\ast(r,y) = 4\pi \int_0^r s^2 \rho^\ast(s)\, ds.
  \]
\end{lemma}
For $\kappa$ very large and close to the center the behavior of the
massive steady state is indeed captured by the massless one, more precisely:
\begin{lemma} \label{p-pml}
  There exists a constant $C>0$ 
  such that for all $\kappa > 0$ and $r\geq 0$,
  \[
  |y_\kappa(r) - y^\ast_\kappa (r)|
  \leq
  C e^{2\kappa} \left(r^2 + e^{4\kappa} r^4 \right)
  \exp\left(C\left(e^{4\kappa} r^2 + e^{8\kappa} r^4\right)\right) .
  \]
\end{lemma}
In \cite[Lemma~3.10]{HaLinRe} this result is proven for the pressures
$\sigma$ instead of the functions $y$, because that allows one to
treat the Einstein-Vlasov and the Einstein-Euler cases simultaneously.
The result above is actually more easy to obtain.
Essentially, the proof consists of a lengthy Gronwall-type
estimate, based on the equations satisfied by $y_\kappa$ and $y^\ast_\kappa$,
but rewritten in the rescaled radial variable $\tau= e^{2\kappa} r$,
and using the facts that due to the Buchdahl inequality
\cite{Andr2007,Andr2008},
\[
\frac{2 m(r)}{r}, \frac{2 m^\ast (r)}{r} < \frac{8}{9},
\]
and that the asymptotics in \eqref{gdef_as}, \eqref{hdef_as} take the
quantitative form
\[
|g(y) - g^\ast(y)| + |h(y) - h^\ast(y)| \leq C e^{2y},\ y\in \R,
\]
This is indeed a good approximation
for large $y$, since all the terms on the left are then of order $e^{4y}$.

Since $g^\ast$ and $h^\ast$ are strictly positive,
a steady state of the massless system is never compactly supported.
The massless system has a scaling invariance which is important for
what follows.
\begin{lemma} \label{z-xi}
  Let $y^\ast_0$ denote the solution of \eqref{grzeq}
  with initial data $y^\ast_0(0)=0$.
  Then for all $\kappa >0$,
  \[
  y^\ast_\kappa(r) = \kappa + y^\ast_0(e^{2\kappa} r),\ r\geq 0 .
  \]
\end{lemma}
We see that in order to understand the behavior of $y^\ast_\kappa(r)$ for
positive, small $r$ and very large $\kappa$ we need to understand the behavior
of the special solution $y^\ast_0(s)$ for $s\to \infty$.
The key point here is that the massless steady state equation~\eqref{grzeq}
can be turned into a planar, autonomous dynamical system.
We let $w_1(\tau) = r^2 \rho(r),\ w_2(\tau)= m(r)/r$ with $\tau=\ln r$.
Then the Tolman-Oppenheimer-Volkov equation \eqref{tov} and the relation between
$\rho$ and $m$ imply that
\begin{align}
  \frac{d w_1}{d\tau}
  &=
  \frac{2 w_1}{1 - 2 w_2}\left(1 - 4 w_2 -\frac{8\pi}{3} w_1 \right),
  \label{w1}\\
  \frac{d w_2}{d\tau}
  &=
  4 \pi w_1 - w_2.\label{w2}
\end{align}
The system has two steady states,
\[
(0,0) \ \mbox{and}\
Z\coloneqq\left(\frac{3}{56\pi},\frac{3}{14}\right).
\]
Using Poincar\'{e}-Bendixson theory it can be shown that
there is a unique trajectory which corresponds to one branch $T$ of the unstable
manifold of $(0,0)$ and converges to $Z$ with a rate determined by
the real parts of the eigenvalues of the linearization at $Z$, which equal
$-\frac{3}{2}$.
For the solution induced by $y^\ast_0$ it holds that $w(\tau) \to (0,0)$ for
$\tau\to -\infty$, and its trajectory coincides with $T$.
The result is that
for any $0<\gamma < 3/2$ and all $\tau$ sufficiently large, 
\[
|w(\tau) - Z| \leq C e^{-\gamma \tau}.
\]
When we rewrite this in terms of the original variables and combine it with
the previous three lemmata, we obtain the following result: 
\begin{proposition}\label{P:EVTOBKZ}
  There exist parameters $0<\alpha_1 < \alpha_2 < \frac{1}{4}$, $\kappa_0>0$
  sufficiently large, and constants $\delta>0$ and $C>0$
  such that on the {\em critical layer}
  \[
  [r_\kappa^1, r_\kappa^2] \coloneqq
  [\kappa^{\alpha_1} e^{-2\kappa},\kappa^{\alpha_2} e^{-2\kappa}]
  \]
  and for every $\kappa\ge \kappa_0$
  the following estimates hold:
  \[
  \left|r^2 \rho_\kappa (r) - \frac{3}{56 \pi}\right|,\,
  \left|r^2 \sigma_\kappa (r) - \frac{1}{56 \pi}\right|,\,
  \left|\frac{m_\kappa(r)}{r} -\frac{3}{14}\right|,\,
  \left|2r\mu_\kappa'-1\right|,\,
  \left|e^{2\lambda_\kappa}-\frac{7}{4}\right|,\,  \left|r \lambda_\kappa'\right|
  \leq
  C\,\kappa^{-\delta}.
  \]
\end{proposition}
For more details on the proof of this result we refer to
\cite[Props.~3.13, 3.14]{HaLinRe},
but we wish to discuss the limiting object which corresponds to the stationary
state $Z$ of the dynamical system \eqref{w1}, \eqref{w2}.
Indeed, $Z$ corresponds to the macroscopic data
\[
\rho(r) = \frac{3}{56 \pi} r^{-2},\ \sigma(r) = \frac{1}{56 \pi} r^{-2},\
m(r) = \frac{3}{14} r,\ \frac{2m(r)}{r} =  \frac{3}{7},\
e^{2\lambda} = \frac{7}{4},\ \mu'(r) = \frac{1}{2 r},
\]
which represent a particular steady state of the massless (EV) system; there is
a free constant when defining $\mu$ which we take such that
$e^{2\mu(r)} = \frac{7}{4} r$.
We refer to this solution as the BKZ solution, because these
macroscopic quantities are the same as for a certain massive solution
found by {\sc  Bisnovatyi-Kogan} and {\sc Zel'dovich} in \cite{BKZ}. 
It does not represent a regular, isolated system:
it violates both the condition \eqref{boundc0} for a regular center and
for asymptotic flatness \eqref{boundcinf}, and it
has infinite mass. 
Its Ricci scalar vanishes, while its Kretschmann scalar
\[
K(r) \coloneqq R_{\alpha\beta\gamma\delta} R^{\alpha\beta\gamma\delta}(r)
= \frac{72}{49} r^{-4}
\]
blows up at the center;
the BKZ solution has a spacetime singularity at $r=0$.
The curves
\[
r(t) = (c + t/2)^2,\
t > - 2c
\]
with $c>0$ represent radially outgoing null geodesics which start
at the singularity and escape to $r=\infty$, i.e.,
the singularity is visible for observers away from the singularity.
Hence it violates the strong cosmic censorship hypothesis;
the concept of weak cosmic censorship is not applicable to this solution,
since it is not asymptotically flat. According to the
cosmic censorship hypothesis such ``naked'' singularities
should be ``non-generic'' and/or ``unstable''.
The analysis which we review in the present section shows
that regular steady states, which in the critical layer are
close to the BKZ solution
for large central redshift $\kappa$, seem to inherit this instability and are
indeed unstable themselves.

We also point out that the BKZ solution can for obvious reasons not capture
the behavior of the massive (EV) steady state at the center or for
large radii. But the information provided in Prop.~\ref{P:EVTOBKZ}
on the critical layer $[r_\kappa^1, r_\kappa^2]$ turns out to be what is needed
for the next step.
\subsection{A negative energy direction for $\kappa$ large}
When $\kappa$ is sufficiently large there exists a
linearly dynamically accessible direction in which
the second variation of $\H_C$, i.e.,
the bilinear form $\mathcal A$, becomes negative.
\begin{theorem}\label{T:NEGATIVEA}
  There exists $\kappa_0>0$ such that for all $\kappa>\kappa_0$
  there exists a spherically symmetric, odd-in-$v$ function
  $h\in C^2(\R^6)$ such that
  \[
  \mathcal A (h,h) = \mathcal A (f_h,f_h) < 0,
  \]
  where $f_h$ is given by~\eqref{dynacdef}.
\end{theorem}
The negative energy direction $h$ is of the form
\begin{equation}\label{E:HG}
  h(x,v) = g(r) w,
\end{equation}
with a suitable function $g\in C^2([0,\infty[)$.
Clearly, $h$ is spherically symmetric and odd in $v$,
where we recall \eqref{rwL_ev}.
A suitable integration by parts implies that
\begin{equation}\label{intpar1}
  \int \phi'(E)\,w^2\,dv
  = - e^{-\mu_\kappa}(\rho_\kappa + \sigma_\kappa).
\end{equation}
Combining this with \eqref{E:HG} and  \eqref{laplusmu} the expression
\eqref{ldynaclambda} for $\lambda_h$
can be simplified:
\[
\lambda_{h} = 4\pi r  e^{\mu_\kappa+\lambda_\kappa}g \int \phi'(E)\,w^2\,dv 
= - e^{-\lambda_\kappa}g\left(\lambda_\kappa'+\mu_\kappa'\right).
\]
Moreover,
\[
\{h,f_\kappa\}
= \phi'(E)\, e^{\mu_\kappa}
\left( g'(r)\frac{w^2}{\pv} - \mu_\kappa'g(r)\pv
+ g(r)\frac{|v|^2-w^2}{r\pv}\right),
\]
and thus
\begin{align*}
  f_{h} = e^{\mu_\kappa-\lambda_\kappa}
  \phi'(E)\left((g'-g(\mu_\kappa'+\lambda_\kappa'))
  \frac{w^2}{\pv}-\mu_\kappa'g\pv + g\frac{|v|^2-w^2}{r\pv}\right).
\end{align*}
On the critical layer $[r_\kappa^1, r_\kappa^2]$ the steady state
$(f_\kappa,\lambda_\kappa,\mu_\kappa)$ is well approximated by the BKZ
solution, provided $\kappa$
is sufficiently large, cf.\ Prop.~\ref{P:EVTOBKZ}. 
We localize the perturbation $h$ given by \eqref{E:HG}
to this interval by setting
\[
g = e^{\frac12\mu_\kappa+\lambda_\kappa}\chi,
\]
where $0\leq \chi\leq 1$
is a smooth cut-off function supported in the interval
$[r_\kappa^1, r_\kappa^2]$ and equal to $1$ on
$[2 r_\kappa^1, r_\kappa^2/2]$; the latter interval is non-trivial
for $\kappa$ sufficiently large. In addition, $|\chi'(r)|$
is to satisfy certain bounds which are not relevant here.
The perturbation $f_{h}$ now takes the form
\[
f_{h} = e^{\frac32\mu_\kappa} \phi'(E)
\left(-\mu_\kappa' \chi\left[\frac{w^2}{2\pv}
  + \pv -\frac{1}{\mu_\kappa'r}\frac{|v|^2-w^2}{\pv}\right]+
\chi'\frac{w^2}{\pv}\right).
\]
Substitution of this expression into~\eqref{linen}
yields the following identity:
\begin{align*}
  \mathcal{A}(h,h) 
  &= \int_{r_\kappa^1}^{r_\kappa^2} e^{2\mu_\kappa-\lambda_\kappa} \chi^2
  \biggl[ 4\pi r^2 e^{\mu_\kappa+2\lambda_\kappa}(\mu_\kappa')^2
  \int |\phi'(E)|
  \left(\frac{w^2}{2\pv} + \pv
  -\frac{1}{\mu_\kappa'r}\frac{|v|^2-w^2}{\pv}\right)^2\,dv \\
  &
  \qquad \qquad \qquad \qquad  - (2r\mu_\kappa'+1)
  (\mu_\kappa'+\lambda_\kappa')^2\biggr]\,dr \\
  &
  \quad
  + 4\pi \int_{\text{supp}\chi'} r^2 e^{3\mu_\kappa+\lambda_\kappa}(\chi')^2
  \int |\phi'(E)|\frac{w^4}{\pv^2}\,dv\,dr \\
  &
  \quad
  + 8\pi \int_{\text{supp}\chi'}
  r^2 e^{3\mu_\kappa+\lambda_\kappa} \mu_\kappa' \chi \chi'
  \int |\phi'(E)| \left(\frac{w^4}{2\pv^2}+w^2
  -\frac{1}{\mu_\kappa'r}\frac{w^2|v|^2-w^4}{\pv^2}\right)\,dv\,dr.
\end{align*}
The key point now is that if in the first integral
the steady state quantities are replaced by their corresponding
limiting quantities according to Prop.~\ref{P:EVTOBKZ},
a strictly negative term arises, together with error
terms, which, being like the second and third integral of
lower order in $\kappa$, do not destroy the negative
sign of $\mathcal{A}(h,h)$, provided $\kappa$ is sufficiently
large; for the details we have to refer to \cite[Theorem~4.3]{HaLinRe}
\footnote{The proof published in \cite{HaLinRe}
          contains an error which has been corrected in
          \href{https://arxiv.org/abs/1810.00809}{arXiv:1810.00809}.}.
\subsection{Linear exponential instability}
An adaptation of an argument by
{\sc Laval, Mercier, Pellat}~\cite{LaMePe} shows
that the existence of a negative energy direction as provided by
Theorem~\ref{T:NEGATIVEA} implies a linear exponential instability result.
At first glance this may seem surprising, since the energy $\mathcal A$
could be negative definite in which case its conservation should
imply stability. In order to appreciate the role of
the negative energy direction, $h$ in \eqref{heqn} must be split
into even and odd parts with respect to $v$,
which turns the latter equation into the system
\begin{align}
  \partial_t h_- + \mathcal T h_+ & = 0,\label{hsystem1}\\
  \partial_t h_+  + \mathcal T h_- & =  \mathcal C h_- .\label{hsystem2}
\end{align}
Here
\[ 
\mathcal T h \coloneqq e^{-\lambda_\kappa}\{h,E\},\qquad
\mathcal C h \coloneqq
-e^{\mu_\kappa}\lambda_{h} \frac{w^2}{\pv} + e^{\mu_\kappa} \mu_{h} \pv.
\]
Let $L^2_{W}=L^2_{W}(D)$ denote the weighted $L^2$ space on the set
$D=\{f_0>0\}$
with the weight
$W \coloneqq e^{\lambda_\kappa}|\phi'(E)|$,  and let
$\langle\cdot,\cdot\rangle_{L^2_{W}}$ denote the corresponding scalar product.
As we noted before, solutions to  \eqref{heqn} conserve the energy
$\mathcal A (h,h)$.
But substituting $h=h_+ + h_-$ and using
the fact that $\lambda_h = \lambda_{h_-}$ it follows that
\[
\mathcal A(h,h) = \mathcal A(h_-,h_-)
+ \langle\mathcal T h_+,\mathcal T h_+\rangle_{L^2_{W}}.
\]
Hence for the system
\eqref{hsystem1}, \eqref{hsystem2} conservation of energy takes the form
\[ 
\langle\mathcal T h_+,\mathcal T h_+\rangle_{L^2_W} + \mathcal A (h_-,h_-)
= const,
\]
and $\mathcal A$ now plays the role
of potential energy. A negative direction for the latter together with the
positive definiteness of the kinetic energy gives a saddle point structure
for the total energy, and instability is expected.

Using the fact that solutions of the system
\eqref{hsystem1}, \eqref{hsystem2} also satisfy the virial identity
\[ 
\frac{1}{2}\frac{d^2}{dt^2}\langle h_-,h_-\rangle_{L^2_{W}} =
- \mathcal A(h_-,h_-) + \langle\mathcal T h_+,\mathcal T h_+\rangle_{L^2_{W}}
\]
one can now follow the idea in \cite{LaMePe} to derive
the following linear, exponential instability result; for details we refer to
\cite[Theorem~4.8]{HaLinRe}.
\begin{theorem} \label{lininstabEV}
  There exist initial data
  $\mathring h_+, \mathring h_-\in C^1(\overline{D})$ and constants $c_1,c_2>0$
  such that for the corresponding solution to the
  system \eqref{hsystem1}, \eqref{hsystem2},
  \[
  \|h_-(t)\|_{L^2_{W}} ,  \|\mathcal T h_+(t)\|_{L^2_{W}}
  \ge c_1 e^{c_2 t}.
  \]
\end{theorem}
A much stronger result, namely the existence of an exponentially
growing mode, is discussed in the next section.
\section{Spectral properties of the linearized (EV) system}
\label{FA}
\setcounter{equation}{0}
\subsection{The functional-analytic structure of the linearized (EV) system}
Important aspects of the linearized (EV) system such as the existence
of exponentially growing modes for $\kappa$ sufficiently large
can only be properly understood, if the linearized system is put into
the proper functional-analytic framework. The latter is set up on
the real Hilbert space
\[ 
H \coloneqq \left\{f\colon D \to \R
\text{ measurable and spherically symmetric}
\mid \|f\|_H < \infty \right\},  
\]
where the norm $\|f\|_H$ is defined in terms of the scalar product
\[ 
\langle f,g\rangle_{H} \coloneqq
\iint_{D} \frac{e^{\lambda_0}}{|\phi'(E)|}\,f\, g \, dv\, dx, \quad g,h\in H;
\]
for the moment we consider some fixed steady state $(f_0,\lambda_0,\mu_0)$
and ignore the dependence on the central redshift $\kappa$;
we recall that $D = \{f_0>0\}$.
We need to define the transport operator
$\T f = e^{-\lambda_0}\{f,E\}$ where, say, $f\in C^1(D)$, as an operator on $H$.

We say that for a function $f\in H$ the {\em transport term} $\T f$
{\em exists weakly} if there exists $h \in H$ such that for every
spherically symmetric test function $\xi \in C^1_{c}(D)$,
\[ 
\langle f, \T \xi \rangle_H = - \langle h, \xi \rangle_H. 
\]
If such a function $h$ exists, it is unique, and
we set {\em $\T f = h $ in a weak sense}.
The domain of $\T$ is defined as 
\[ 
\Ds(\T) \coloneqq \{f \in H \, | \, \T f \in H \text{ exists weakly}\},
\]
and the resulting operator $\T\colon \mathrm \Ds(\T)\to H$ is the
{\em transport operator}. In view of \eqref{linvlasov} we also define
$\B\colon  \Ds(\T) \to H$ by
\begin{equation}\label{Bdef} 
  \B f \coloneqq - \T f - 4\pi r |\phi'| e^{2\mu_0+\lambda_0}
  \left (w \sigma_f - \frac{w^2}{\pv}\,\jmath_f \right ),
\end{equation}
and the {\em residual operator} $\mathcal{R}\colon H \to H$ by 
\[ 
\mathcal{R}f \coloneqq 4\pi |\phi'|\,e^{3\mu_0} (2r\mu_0'+1) w \jmath_f.
\]
These operators have the following properties
\footnote{In the literature the sign in front of $\T$ is not always chosen
          as consistently as we try to.}:
\begin{lemma}\label{oplemma1}
  \begin{itemize}
  \item[(a)]
    The transport operator $\T \colon \Ds(\T) \to H$ is densely defined
    and skew-adjoint, i.e.,
    $\T^\ast=-\T$, and $\T^2 \colon \mathrm \Ds(\T^2) \to H$
    with
    \[
    \Ds(\T^2) \coloneqq \left\{f \in H \, |\, f \in \Ds(\T), \,
    \T f \in \Ds(\T) \right\}
    \]
    is self-adjoint.
  \item[(b)]
    The operator $\B \colon \Ds(\T) \to H$ is densely defined and skew-adjoint,
    and $\B^2 \colon \Ds(\T^2) \to H$ is self-adjoint.
  \item[(c)]
    The operator $\Ro \colon H \to H$ is bounded, symmetric,
    and non-negative, i.e., $ \langle \Ro f, f \rangle_H \geq 0$ for $f\in H$.
  \end{itemize}
\end{lemma}
That the transport operator is symmetric with respect to the scalar product
on the Hilbert space $H$ is easy to see; for the details of the above results
we refer to \cite{GueStrRe} or \cite{HaLinRe}. We use these operators to put
the linearized (EV) system, i.e., \eqref{linvlasov}, into the form
\begin{equation}\label{linev}
  \partial_t f = \B f 
  - e^{2\mu_0-\lambda_0} \left(2\mu_0'+\frac{1}{r}\right)
  \lambda_{f} |\phi'(E)|\, w;
\end{equation}
note that we simply write $f$ instead of $\delta f$ here
and in what follows.
As before and following {\sc Antonov} we split $f=f_+ + f_-$ into
its even and odd parts with respect to $v$. Since $\B$ reverses $v$-parity,
\begin{align*}
  \partial_t f_+
  &=
  \B f_-,\\
  \partial_t f_-
  &=
  \B f_+ 
  - e^{2\mu_0-\lambda_0} \left(2\mu_0'+\frac{1}{r}\right)
  \lambda_{f} |\phi'(E)|\, w.
\end{align*}
Differentiating the second equation with respect to $t$ and substituting the
first one implies that
\begin{align*}
  \partial^2_t f_-
  &=
  \B f_+- e^{2\mu_0-\lambda_0} \left(2\mu_0'+\frac{1}{r}\right)
  \partial_t \lambda_{f} |\phi'(E)|\, w\\
 &=
  \B^2 f_- + 4 \pi e^{3\mu_0} \left(2 r \mu_0'+1\right)
  \jmath_{f_-} |\phi'(E)|\, w\\
  &=
  \B^2 f_- + \Ro f_-,
\end{align*}
where we used \eqref{linein3}, the fact that $\jmath_f = \jmath_{f_-}$ and the
definition of the residual operator $\Ro$.
Since this second-order formulation of the
linearized system lives on the odd-in-$v$ parts of the perturbations, we
define
\[
H^\mathrm{odd} \coloneqq \{f\in H \mid f\ \mbox{is odd in}\ v\},
\]
which is a Hilbert space with the same scalar product as before.
The properties of the operators stated in Lemma~\ref{oplemma1} remain
true on $H^\mathrm{odd}$, and we define the
{\em Antonov operator}
\[
\L \colon \Ds(\L) \to H^\mathrm{odd},\  
\Ds(\L)\coloneqq \Ds(\T^2)\cap H^\mathrm{odd},\ 
\L\coloneqq -\B^2 - \Ro.
\]
This is again a self-adjoint operator, and the linearized (EV) system
is put into the form
\begin{equation}\label{linev_2ndo}
  \partial^2_{t} f_- + \L f_- = 0,
\end{equation}
which has the same structure as the corresponding equation \eqref{linvp_2ndo}
for (VP).

The above second-order formulation has been used in the astrophysics literature,
cf.~\cite{Ip69,Ip69b,IT68} (without precise spaces, domains etc).
Based on \eqref{linev_2ndo} we call
a steady state of the Einstein-Vlasov system
{\em linearly stable} if the spectrum of $\L$ is strictly positive, i.e., 
\[ 
\gamma\coloneqq\inf \sigma(\L) > 0; 
\]
notice that the spectrum of $\L$ is real since $\L$ is self-adjoint.
By \cite[Prop.~5.12]{HiSi} this spectral condition
implies the Antonov-type inequality
\[ 
\langle f, \L f\rangle_H \geq \gamma \|f\|_H^2,\
f\in D (\L).
\]
Since
\[ 
\|\partial_t f_-\|_H^2 + \langle f_-, {\cal L} f_-\rangle_H
= \A(\partial_t f_-,\partial_t f_-)
\]
is conserved along solutions of the linearized equation	\eqref{linev_2ndo},
this implies linear stability in the corresponding norm.

Assume on the other hand that $\alpha<0$ is an eigenvalue of $\L$
with eigenfunction $f \in H^\mathrm{odd}$.
Then $g \coloneqq e^{\sqrt{-\alpha}\, t} f$  solves \eqref{linev_2ndo}, 
and we get an exponentially growing solution of the linearized (EV) system.
Hence an eigenfunction  $f \in H^\mathrm{odd}$ to a
negative eigenvalue $\alpha<0$ of $\L$ is called an
{\em exponentially growing mode}.
Using Theorem~\ref{T:NEGATIVEA} one can show that
for $\kappa$ sufficiently large, such exponentially growing modes exist.

To see this one needs some further tools which are also used to obtain
a first-order formulation of the linearized (EV) system with good
functional-analytic properties. This first-order formulation has to our
knowledge not appeared in the physics literature and was introduced in
\cite{HaLinRe}.
A key ingredient is a modified potential
induced by a state $f\in H$:
\begin{equation}\label{mubardef}
\bar\mu(r) = \bar\mu_f(r)
\coloneqq - e^{-\mu_0-\lambda_0}\int_r^\infty \frac1s\,
e^{\mu_0(s)+\lambda_0(s)}(2s\mu_0'(s)+1)\,\lambda_f(s)\,ds
\end{equation}
is the {\em modified potential induced by $f\in H$}, where
$\lambda_f$ is defined by \eqref{linein1}. It has the following properties,
where $\dot H^1_r$ denotes the subspace of spherically symmetric functions in
the homogeneous Sobolev space $\dot H^1(\R^3)$, cf.~\cite{Ge1998}.
\begin{lemma} \label{mubarlemma}
  \begin{itemize}
  \item[(a)]
    For $f\in H$,
    $\bar\mu=\bar\mu_f \in C([0,\infty[)\cap C^1(]0,\infty[)\cap \dot H^1_r$,
    and $|\bar \mu (r)| \leq C \| f\|_H$, $r\geq 0$, with some $C>0$ independent
    of $f$.
  \item[(b)]
    It holds that
    \begin{equation} \label{mubarder}
      \bar \mu' = - (\mu_0' + \lambda_0') \, \bar \mu
      + \frac{2r\mu_0' +1}{r} \lambda_f,
    \end{equation}
    \begin{equation}\label{mubareq1}
      \frac{e^{-\mu_0-\lambda_0}r}{2r\mu_0'+1}(e^{\mu_0+\lambda_0} \bar\mu)'
      = \lambda_f,
      \ r\geq 0,
    \end{equation}
    and in the weak sense,
    \begin{equation}\label{mubareq2}
      \frac{1}{4\pi r^2}
      \frac{d}{dr}\left(\frac{e^{-\mu_0-3\lambda_0}r^2}{2r\mu_0'+1}
      \frac{d}{dr}\left(e^{\mu_0+\lambda_0} \bar\mu\right)\right) =
      \rho_f \ \mbox{a.~e.}
    \end{equation}
  \item[(c)]
    The operator $\K\colon H\to H$, $\K f\coloneqq \phi'(E)E \bar\mu_f$
    is bounded, self-adjoint, and compact.
\end{itemize}
\end{lemma}
One should at this point recall \eqref{muprtermpos}.
The field equation~\eqref{ein2} and the boundedness of
$\sigma_0,\mu_0,\lambda_0$ imply that the quantity
$e^{\mu_0+\lambda_0}(2 r \mu_0'+1)$ is bounded.
The estimate for $\bar\mu$ then follows by the Cauchy-Schwarz inequality.
For the remaining assertions one should observe that
by \eqref{laplusmu}, $\mu_0'+\lambda_0'=0$ outside $D$, and
that $r^2 \rho\in L^1([0,\infty[)$. That $\K$ is bounded follows from part~(a),
integration-by-parts and \eqref{mubareq2} imply its self-adjointness,
and compactness follows using the the Arzela-Ascoli theorem, where it
is important that the steady state has compact radial support.
For the details we refer to \cite[Lemmata~4.17, 4.18]{HaLinRe}.

The compactness of the map $\K$ is important for the operator
$\bar{\mathcal L} \colon H \to H$ defined by
\begin{align}
  \bar{\mathcal L} f & \coloneqq
  f - \phi'(E)E \bar\mu_f. \label{barLdef}
\end{align}
By Lemma~\ref{mubarlemma}~(c):
\begin{lemma}
  The operator $\bar{\mathcal L}$
  is bounded and symmetric on $H$.
\end{lemma}
The linearized (EV) system can now be put into
the following first order Hamiltonian form
which means that the general theory
developed in \cite{LinZeng2017} can be applied.
\begin{proposition} \label{L:EVLIN1}
  The linearized (EV) system takes the form
  \begin{align}\label{linev_1sto}
    \partial_t f = \mathcal B \bar{\mathcal L} f,
  \end{align}
  $\Ds(\mathcal B \bar{\mathcal L})=\Ds(\mathcal T)$,
  the operator $\bar{\mathcal L}$ induces the
  quadratic form 
  \begin{equation}\label{qformbarL}
    \langle\bar{\mathcal L} f,f\rangle_H
    = \iint \frac{e^{\lambda_{0}}}{|\phi^{\prime }(E)|}
    f^2\,dv\,dx
    - \int_0^\infty e^{\mu_0-\lambda_0}
    \left( 2r\mu_0'+1\right) \,\lambda_f^2\,dr = \mathcal A(f,f)
  \end{equation}
  on $H$, and the flow of \eqref{linev_1sto} preserves $\mathcal A(f,f)$.
  The relation of the first-order formulation \eqref{linev_1sto}
  to the second-order one in \eqref{linev_2ndo} is captured in the relation
  \begin{equation} \label{LBrel}
    \L = -\B \bar\L \B.
  \end{equation}
\end{proposition}
We refer to \cite[Lemma~4.20]{HaLinRe} for a rigorous proof
and highlight only some instructive aspects.
For $f\in H$, at least formally,
\[ 
\mathcal T (\phi'(E) E \bar\mu_f)
= e^{2 \mu_0-\lambda_0} \phi'(E)\, w\, \bar\mu_f'.
\]
Together with \eqref{intpar1} and \eqref{mubarder} this implies that
\[
\mathcal{B}\left(
\phi^{\prime }E \bar{\mu}_f\right)=
-\phi^{\prime }(E)e^{2\mu_0-\lambda_0}w
\frac{2r\mu_0^{\prime }+1}{r}\lambda_f.
\]
If we combine this with the form \eqref{linev} of the linearized (EV) system
we obtain
\[
\partial _{t}f=\mathcal{B} \left( f-\phi^{\prime}(E)
E \bar{\mu}_f\right) = \mathcal{B} \bar{\mathcal L} f.
\]
If we differentiate $\langle\bar{\mathcal{L}} f,f\rangle_H$
with respect to $t$
and use the symmetry of $\bar{\mathcal L}$, \eqref{linev_1sto},
and the skew-adjointness of $\mathcal B$ the conservation law follows.
By the definition of $\bar{\mathcal L}$,
\begin{align*}
  \langle\bar{\mathcal{L}} f,f\rangle_H & =
  \langle f-\phi^{\prime }(E) E \bar{\mu_f},\ f\rangle_H \\
  & = \iint \frac{e^{\lambda_0 }}{|\phi'(E)|} f^{2}\,dv\,dx
  + \int e^{\mu_0+\lambda_0} \bar\mu_f \int f\pv\,dv\,dx \\
  & =  \iint \frac{e^{\lambda_0}}{|\phi'(E)|} f^{2}\,dv\,dx
  + 4\pi \int_0^\infty r^2\bar\mu_f\,  e^{\mu_0+\lambda_0}\rho_f\,dr,
\end{align*}
and since $4\pi r^2 \rho_f = \left(e^{-2\lambda_0}r\lambda_f\right)'$
the assertion \eqref{qformbarL} follows by using~\eqref{mubareq1}.

If we split some element $f\in \Ds(\B \bar\L)=\Ds(\T)$
into its even and odd parts with respect to $v$, it follows that
$\lambda_f =\lambda_{f_+}$, hence also $\bar\mu_f=\bar\mu_{f_+}$,
and $\bar\L f = \bar\L f_+ + f_-$. Since $\bar\L$ preserves $v$ parity
and $\B$ reverses it, the first order formulation \eqref{linev_1sto}
splits into
\[
\partial_t f_+ = \B f_-,\ \partial_t f_- = \B \bar\L f_+
\]
which directly implies
\[
\partial_t^2 f_- = \B \bar\L \B f_-
\]
as desired; of course the relation \eqref{LBrel} can be checked directly.

The spectral properties of the operators $\mathcal L$ or
$\mathcal B \bar{\mathcal L}$ are difficult to analyze,
and a key idea to do so is to find simpler,
macroscopic Schr\"odinger-type operators
by which for example $\mathcal L$
is bounded from above and below.
These reduced operators act on functions of only the radial variable $r$, 
which makes them easier to analyze.

The construction which we explain below was developed in
\cite{HaLinRe} and relies on the modified
potential $\bar\mu_f$ as a key ingredient.
An earlier, but not really satisfactory attempt
to construct such a reduced operator was made in~\cite{Ip69}. 

The {\em modified Laplacian} $\bar \Delta$ is given by
\[ 
\bar \Delta \psi \coloneqq
\frac{e^{\mu_0+\lambda_0}}{4\pi r^2}
\frac{d}{dr}\left(\frac{e^{-\mu_0-3\lambda_0}r^2}{2r\mu_0'+1}
\frac{d}{dr} \left(e^{\mu_0+\lambda_0}\psi\right)\right).
\]
On a flat background, i.e., for $\lambda_0=\mu_0=0$ the
operator $ 4\pi \bar \Delta$ is the Laplacian applied to
spherically symmetric functions.
The {\em reduced operator} $S$ is given by
\begin{equation}\label{Sdef}
  S \psi \coloneqq 
  - \bar \Delta \psi
  - e^{\lambda_0}\int |\phi'(E)|E^2\,dv\,\psi,
\end{equation}
and the {\em non-local reduced operator}
$\tilde S$ is
\begin{equation}\label{Stildedef}
  \tilde S \psi \coloneqq 
  - \bar \Delta \psi
  - e^{\lambda_0}\int \left(\mathrm{id}-\Pi\right)
  \left(|\phi'(E)|\,E \psi \right) E\,dv,
\end{equation}
where $\Pi$ denotes the projection onto $\Rs(\mathcal B)^\perp$, the orthogonal
complement in $H$ of the range of the operator $\mathcal B$,
and $\mathrm{id}$ is the identity.

For what follows, $\bar \mu_f$, which only belongs to $\dot H^1_r$,
must lie in the domain of $S$ and $\tilde S$.
Hence one must be careful to define these operators between
the proper spaces, which is done using duality;
$(\dot H^1_r)'$ denotes the dual space of $\dot H^1_r$
and $\langle\cdot,\cdot\rangle$ denotes the corresponding
dual pairing.
\begin{lemma} \label{redEVdef}
   The operator $S \colon \dot H^1_r \to  (\dot H^1_r)'$ defined by
  \begin{align*}
    \langle S \psi,\chi\rangle
    \coloneqq & \int_0^\infty \frac{e^{-\mu_0-3\lambda_0}}{2r\mu_0'+1}
    \frac{d}{dr} \left(e^{\mu_0+\lambda_0}\psi\right)
    \frac{d}{dr} \left(e^{\mu_0+\lambda_0}\chi\right)r^2\,dr
    - \iint_{D} e^{\lambda_0} |\psi'(E)|E^2 \psi \chi\,dv\,dx
  \end{align*}
  is self-dual, and it is given by \eqref{Sdef} on
  sufficiently regular functions.
  The operator $\tilde S\colon \dot H^1_r \to  (\dot H^1_r)'$
  is defined analogously and has the analogous properties.
\end{lemma}
These operators bound the quadratic form $\A$ which
according to \eqref{qformbarL} is induced by $\tilde\L$
from above and below in the following precise sense:
%
%
\begin{proposition}\label{ASest}
\begin{enumerate}
\item[(a)]
  For every $\mu\in \dot H^1_r$ and
  $f = f_\mu \coloneqq \phi_\kappa'(E_\kappa)E_\kappa\mu$,
  \[ 
  \langle S\mu,\,\mu\rangle \ge {\mathcal A}( f_\mu, f_\mu).
  \]
  For every $f\in H$ and $\bar\mu_f$ as defined in \eqref{mubardef},
  \[ 
  {\mathcal A}( f, f) \ge \langle S\bar\mu_f,\,\bar\mu_f\rangle.
  \]
\item[(b)]
  For every $\mu\in \dot H^1_r$ and
  $\tilde f = \tilde f_\mu \coloneqq
  (\mathrm{id}-\Pi)(\phi'(E)E \mu)\in \Rs(\mathcal B)$,
  \[
  \langle \tilde S\mu,\,\mu\rangle \ge
  {\mathcal A}( \tilde f_\mu, \tilde f_\mu).
  \]
  For every $f\in H$,
  \[ 
  {\mathcal A}( f, f) \ge
  \langle\tilde S\bar\mu_f,\,\bar\mu_f\rangle.
  \]
\end{enumerate}
\end{proposition}


The proof relies on the observations that for $\mu\in \dot H^1_r$,
\[ 
\langle-\bar \Delta \mu,\mu\rangle
= \int_0^\infty
\frac{e^{-\mu_0-3 \lambda_0}}{2r\mu_0'+1}
\left(\left(e^{\mu_0+\lambda_0} \mu\right)'\right)^2 r^2 \,dr,
\]
in particular, for $f\in H$,
using \eqref{mubardef} and \eqref{mubareq1},
\begin{align*} 
  \langle-\bar \Delta \bar\mu_f,\bar\mu_f\rangle
  &=
  \int_0^\infty
  e^{\mu_0-\lambda_0} \left(2r\mu_0'+1\right)
  \left(\frac{e^{-\mu_0-\lambda_0}}{2r\mu_0'+1} r
  \left(e^{\mu_0+\lambda_0} \bar\mu_f\right)'\right)^2  \,dr \\
  &=
  \int_0^\infty e^{\mu_0-\lambda_0} \left(2r\mu_0'+1\right)\, \lambda_f^2 dr.
\end{align*}
On the other hand by \eqref{mubareq2},
\[
\langle-\bar \Delta \bar\mu_f,\bar\mu_f\rangle
=
- 4\pi \int_0^\infty e^{\mu_0+\lambda_0} \bar\mu_f\, \rho_f\, r^2 dr
= -\iint e^{\mu_0+\lambda_0} \bar\mu \pv f \,dv\,dx.
\]
The definitions of ${\mathcal A}$ and the operators $S$
and $\tilde S$ lead to the desired results; for details
cf.~\cite[Theorem~4.24]{HaLinRe}.
\subsection{(In)stability for the linearized (EV) system}
The results of the previous section imply the existence of an exponentially
growing mode when $\kappa$ is large enough, more precisely:
%
%
\begin{theorem}\label{gromo_thm}
  For $\kappa$ sufficiently large, there
  exists at least one negative eigenvalue of $\L$ and therefore an
  exponentially growing mode
  for the linearization \eqref{linev_2ndo} of (EV) around
  $(f_\kappa, \lambda_\kappa,\mu_\kappa)$; such a steady state is (linearly)
  unstable.  For general $\kappa>0$,
  the negative part of the spectrum of
  $\mathcal L$
  is either empty or consists of at most finitely many eigenvalues
  with finite multiplicities. 
\end{theorem}

Before we sketch the proof we need to introduce some more notation.
Let $L\colon H\supset \Ds(L) \to H$
be a linear, self-adjoint operator on some
Hilbert space $H$.
Its {\em negative Morse index} $n^{-}\left(L\right)$
is the maximal dimension of subspaces of $H$ on which
$\langle L\cdot,\cdot\rangle_H <0$.  The analogous terminology
applies to a self-dual operator $L\colon H \to H'$.
\begin{proof}
  The operator $\mathcal L$ is self-adjoint. For $f\in \Ds(\mathcal L_\kappa)$,
  \[
  \langle\mathcal{L}f,f\rangle_H =\langle\bar{\mathcal{L}}
  \mathcal{B}f ,\mathcal{B}f\rangle_H,
  \]
  see \eqref{LBrel}.
  By Prop.~\ref{ASest}, 
  \[
  n^{-}\left( \mathcal{L}\right) \leq
  n^{-}\left( \bar{\mathcal{L}}\right)
  \le n^{-}\left( S\right) <\infty;
  \]
  the general argument behind
  the first two estimates is reviewed in \cite[Lemma~A.1]{HaLinRe}.
  To show that $n^{-}\left( S\right) <\infty$
  is easier than showing this for $\L$ directly, since $S$
  has a much simpler structure and acts on functions of only the
  radial variable; this is the key point in introducing the reduced operators.
  For $\psi \in \dot H^1_r$,
  $\langle S \psi,\psi\rangle \geq C \langle S' \psi,\psi\rangle$,
  where the self-dual operator
  $S' \colon  \dot H^1_r\to (\dot H^1_r)'$ is formally given as
  $S' = -\Delta -V$ with a non-negative, continuous, compactly supported
  potential $V$. This follows from suitable bounds on
  $\lambda_\kappa$ and $\mu_\kappa$. Now the mapping
  $(-\Delta)^{1/2} \colon \dot H^1(\R^3) \to L^2(\R^3)$,
  $\psi\mapsto (2\pi |\xi|\hat\psi)\check{\phantom{\psi}}$ is an
  isomorphism which respects
  spherical symmetry. Passing to
  $\chi = (-\Delta)^{1/2}\psi$ the relation
  $4 \pi \langle S' \psi,\psi\rangle
  = \langle (\mathrm{id}-\K)\chi,\chi\rangle_{L^2}$
  follows.
  Here $\K=(-\Delta)^{-1/2}V(-\Delta)^{-1/2}\colon L^2(\R^3) \to  L^2(\R^3)$
  is compact, since $V$ is bounded and supported on the compact set
  $\bar B_{R}(0)$ with $[0,R]$ the radial support of the steady state,
  and the map
  $\dot H^1(\R^3) \ni f \mapsto {\bf 1}_{\bar B_{R}(0)} f\in L^2(\R^3)$ is compact;
  notice that $h = V(-\Delta)^{-1/2}\chi\in L^1\cap L^2(\R^3)$ so that
  $\hat h\in L^\infty \cap L^2(\R^3)$, and hence $\frac{1}{2 \pi|\xi|} \hat h$
  and its inverse Fourier transform are in $L^2(\R^3)$.
  The spectral properties of compact operators imply that
  $n^-(\mathrm{id}-\K) <\infty$,
  and invoking \cite[Lemma~A.1]{HaLinRe} again it follows that
  $n^{-}\left( S\right) <\infty$.
  
  The assertion on the spectrum of $\mathcal L$ now
  follows from the spectral representation of this operator;
  the argument is discussed in detail in \cite[Prop.~A.2]{HaLinRe}.
  
  Let us now assume that $\kappa$ is large enough
  to apply Theorem~\ref{T:NEGATIVEA}. That theorem provides a negative
  energy direction, i.e., there exists a
  spherically symmetric function $h\in C^2(\R^6)$ which is odd in $v$
  such that $\A(h,h)=\A(f_h,f_h)<0$. Here $f_h$ is the linearly dynamically
  accessible perturbation generated by $h$ according to \eqref{dynacdef}.
  If we compare this relation to the definition \eqref{Bdef}
  of the operator $\B$
  it follows that $f_h = -\B (\phi' h)$. Hence by \eqref{qformbarL},
  \begin{align*}
    0 > \A(h,h)
    &=
    \A(\B (\phi' h),\B (\phi' h))
    =\langle \bar\L \B(\phi' h),\B(\phi' h)\rangle_H\\
    &=
    - \langle \B \bar\L \B(\phi' h),\phi' h\rangle_H
    =  \langle \L (\phi' h),\phi' h\rangle_H;
  \end{align*}
  for the last two equalities notice that by Lemma~\ref{oplemma1}~(b)
  the operator $\B$ is skew adjoint and the relation \eqref{LBrel} holds.
  Hence by definition
  $n^-(\mathcal L)\ge 1$, and invoking the spectral representation of
  this operator again shows that $\L$ has a negative eigenvalue
  $\alpha<0$ of finite multiplicity,
  cf.~\cite[Prop.~A.2]{HaLinRe}.
  Since the operator $\mathcal L$ is non-negative when
  restricted to the subspace all even-in-$v$ functions in $H$,
  eigenfunctions associated to $\alpha$ must be odd-in-$v$,
  and the existence of an exponentially growing mode is established. 
\end{proof}

\begin{remark}
  \begin{itemize}
  \item[(a)]
    In the proof above the following observation concerning
    the concept of linearly dynamic accessibility was important:
    The fact a state $f=f_h$ is linearly dynamically accessible
    and generated by $h$ according to \eqref{dynacdef}
    is equivalent to saying that $f_h = -\B (\phi' h)$.
    This motivates the following generalization of this concept:
    A function $f\in H$ is a
    {\em linearly dynamically accessible perturbation}
    if $f\in \overline{\Rs(\mathcal B)}$.
  \item[(b)]
    The condition \eqref{ldynac2} for linear dynamic accessibility
    requires that
    \[
    \left\langle \chi'(f_0)\, |\phi'(E)|,
    f + e^{\mu_0}\lambda_f |\phi'(E)| \frac{w^2}{\pv}
    \right\rangle_H =0
    \ \mbox{for all}\ \chi\in C^1(\R)\ \mbox{with}\ \chi(0)=0,
    \]
    and it can be shown that $f\in \overline{\Rs(\mathcal B)}$
    satisfies this orthogonality condition which further justifies the
    generalization of the definition of linear dynamic accessibility.
  \item[(c)]
    If the steady state allows the introduction of action-angle
    variables which according to the discussion in Section~\ref{stst_nonrel}
    is true in particular when $\kappa$ is not too large so that
    the condition \eqref{singwell_ev} holds, then
    $H=\Rs(\B) \oplus \Ns(\B)$, cf.~\cite[Prop.~5.9]{GueStrRe},
    and $\Rs(\B)$ is closed. It is not clear if this is true
    in general. 
  \item[(d)]
    An important feature of linear dynamic accessibility was that it is
    preserved under the linearized flow, cf.~\eqref{heqn}. The generalized
    concept shares this property which can be seen as follows. The exponential
    formula for $C_0$ semigroups \cite[Theorem~8.3]{Pazy} shows that
    \[
    e^{t \B\bar \L} f = \lim_{n\to\infty}
    \left(\mathrm{id}-\frac{t}{n}\B\bar\L\right)^{-n} f.
    \]
    If $f\in \Rs(\B)$ then induction shows that each of the $n$-dependent
    functions on the right is again an element of $\Rs(\B)$, and so is the
    limit if we are in the situation where that space is closed.
    If needed this argument extends to the $\overline{\Rs(\mathcal B)}$-case.
  \end{itemize}
\end{remark}
The same arguments as in the proof above show that
\[
n^{-} ( \mathcal{L}) \leq
n^{-} ( \bar{\mathcal{L}}|_{\overline{\Rs(\mathcal{B})}})
\le n^{-} ( \tilde{S}) <\infty.
\]
Based on the first-order formulation \eqref{linev_1sto} these estimates
and the machinery developed above can be used to derive a detailed
picture of the linearized flow for general $\kappa$. The important point here
is that \eqref{linev_1sto} is a linear Hamiltonian PDE in the sense of
\cite{LinZeng2017}. We list some of the key features here and refer to
\cite[Theorem~4.28]{HaLinRe} for details.

The operator $\mathcal{B}\bar{\mathcal{L}}$
generates a $C_{0}$ group $(e^{t\mathcal{B}\bar{\mathcal{L}}})_{t\in \R}$
of bounded linear operators on $H$. The Hilbert space $H$ can
be decomposed into stable, unstable, and center space, i.e.,
\[
H=E^{s}\oplus E^{u}\oplus E^{c}, 
\]
where $E^{u}$ and $E^{s}$ is the linear subspace spanned by the
eigenvectors corresponding to
positive or negative eigenvalues of
$\mathcal{B}\bar{\mathcal{L}}$ respectively. Moreover,
\[
\dim E^{u}=\dim E^{s}=n^{-} ( \tilde{S}) <\infty .
\]
The subspaces $E^{c}$, $E^{u}$, $E^{s}$ are invariant under
$e^{t\mathcal{B}\bar{\mathcal{L}}}$.
If $\tilde{S}>0$, then the
steady state is linearly stable in the sense that there exists a
constant $C>0$ such that for all perturbations
$f\in H$ and all times $t\in \R$,
\begin{equation} \label{C0groupest}
  \left\| e^{t\mathcal{B}\bar{\mathcal{L}}}f\right\|_{H}
  \leq M\left\Vert f\right\Vert_{H}.
\end{equation}
In the Newtonian limit $\kappa\to 0$
the operator $\tilde{S}$
converges to its Newtonian counterpart, which was
proven to be positive e.g.\ in \cite{GuoLin}. 
By~\eqref{C0groupest} one can therefore obtain
linear stability against general initial data in $H$,
which improves Theorem~\ref{linstabthm}.
\subsection{A Birman-Schwinger principle for (EV)}
\label{bsev}
An important step in the previous two sections was to relate
the generator(s) of the linearized (EV) dynamics to some
operator(s) defined on functions which depend only on the
radial variable $r$. An alternative way to do this is the
Birman-Schwinger principle, which for (VP) we discussed
in Section~\ref{ss_linvp}. In \cite{GueStrRe} a Birman-Schwinger
type principle was developed for (EV), and we now discuss the main features
of this approach.

The general aim is to derive a criterion for the existence of negative
eigenvalues of $\L=-\B^2-\Ro$. The tool we consider here is
the {\em Birman-Schwinger operator}
\begin{equation} \label{Qdef}
  Q\coloneqq-\sqrt\Ro\,\B^{-2}\,\sqrt\Ro\colon H^\mathrm{odd}\to H^\mathrm{odd}
\end{equation}
associated to $\L$; notice that we work on the space
$H^\mathrm{odd}$ of odd-in-$v$
functions since the Antonov operator $\L$ governs the evolution
of that part of the perturbation, cf.~\eqref{linev_2ndo}.
We recall Lemma~\ref{oplemma1} for the basic properties of the operators
$\B^2$ and $\Ro$. In order to define the operator $Q$ one first has
to show that the
non-negative operator $\Ro$ has a square root
$\sqrt\Ro\colon H^\mathrm{odd}\to H^\mathrm{odd}$.
Indeed, this operator can be defined explicitly:
\begin{equation}\label{sqrtRdef}
  \sqrt{\Ro}f\coloneqq4\pi \sqrt{r} |\phi'(E)| e^{2\mu_0+\lambda_0}
  \sqrt{\frac{2r\mu_0'+1}{\mu_0'+\lambda_0'}} \, w \jmath_f
\end{equation}
defines a bounded and symmetric operator on $H^\mathrm{odd}$ with the
property that $\sqrt{\Ro}\sqrt{\Ro} = \Ro$; we recall that
$\jmath_f=\int w f dv$ and that by \eqref{laplusmu} the denominator is positive
in the interior of the radial support $[0,R_0]$
of the steady state under consideration,
cf.~\cite[Lemma~5.15]{GueStrRe}.
Secondly, one can show that the operator
$\B^2 \colon \Ds(\T^2) \cap \Ns(\B^2)^\bot  \to \Rs(\B^2)$
is bijective; the inverse of the latter operator cannot be written down
explicitly, which makes its analysis tricky, but it exists,
cf.~\cite[Prop.~5.14]{GueStrRe}.
The key properties of the
Birman-Schwinger operator $Q$ are captured in the following result.
\begin{proposition}\label{Qprop}
  \begin{itemize}
  \item[(a)]
    The Birman-Schwinger operator $Q$ is linear, bounded, symmetric,
    non-negative, and compact.
  \item[(b)]
    The number of negative eigenvalues of $\L$
    counting multiplicities equals the number of eigenvalues $>1$ of $Q$.
  \end{itemize}  
\end{proposition}
Quite some machinery goes into proving this result, and we try
to explain the main points.
First one introduces a family of auxiliary operators
\[
\L_\gamma\coloneqq-\B^2-\frac1\gamma\Ro\colon
\Ds(\T^2)\cap H^\mathrm{odd}\to H^\mathrm{odd},\ \gamma>0.
\]
Since $\B^2|_{H^\mathrm{odd}}$ is self-adjoint and $\Ro$ is bounded and symmetric,
$\L_\gamma$ is self-adjoint by the Kato-Rellich theorem \cite[Thm.~X.12]{ReSi2}.
Now $Q$ is constructed such that
$0$ is an eigenvalue of $\L_\gamma$ if and only if
$\gamma$ is an eigenvalue of $Q$, and the multiplicities of
these eigenvalues are equal:
If $f\in \Ds(\T^2)\cap H^\mathrm{odd}$ solves $\L_\gamma f=0$, i.e.,
$-\gamma\,\B^2f=\Ro\,f$, then applying $-\sqrt\Ro\,\B^{-2}$ to the latter
equation and writing $\Ro=\sqrt\Ro\,\sqrt\Ro$ yields
\[
\gamma\,g = \gamma\,\sqrt\Ro\,f= Q\left(\sqrt\Ro\,f\right) = Qg,
\]
with $g\coloneqq\sqrt\Ro\,f \in H^\mathrm{odd}$. 
The converse direction is similar.

Next, the operator $\Ro|_{H^\mathrm{odd}}$ can be shown to be relatively
$(\B^2|_{H^\mathrm{odd}})$-compact so that by Weyl's theorem~\cite[Thm.~14.6]{HiSi},
\[
\sigma_{ess}(\L_\gamma)= \sigma_{ess}(-\B^2|_{H^\mathrm{odd}}) =\sigma_{ess}(\L).
\]
In addition, one can show that $\sigma(-\B^2|_{H^\mathrm{odd}})$
is positive and bounded away from $0$, and hence $\inf(\sigma_{ess}(\L))>0$.
It remains to understand the behavior of the isolated eigenvalues of $\L_\gamma$
when varying~$\gamma$.
This can be done by their variational characterization,
and one can show that the number of negative eigenvalues
of $\L$ equals the number of $\gamma$'s for which $0$ is an eigenvalue
of $\L_\gamma$.
This establishes the relation
in Prop.~\ref{Qprop}~(b); for details we refer to
\cite[Sections~6.1,~6.2]{GueStrRe}.

By Prop.~\ref{Qprop} the original question of negative eigenvalues of $\L$
is translated into an eigenvalue problem for the Birman-Schwinger operator which
has quite favorable qualities.
But following {\sc Mathur's} idea encountered for (VP) in Section~\ref{ss_linvp}
one can exploit the structure which manifests itself in
\eqref{Qdef}, \eqref{sqrtRdef}
to pass to an even simpler operator.
By  \eqref{Qdef} an eigenfunction of $Q$ which
corresponds to a non-zero eigenvalue lies in the range of $\sqrt\Ro$, and
by \eqref{sqrtRdef},
\[ 
  \Rs(\sqrt\Ro)\subset\left\{f=f(x,v)=\vert\phi'(E)\vert\,w\, \alpha_0(r) \,F(r)
  \text{ a.e.}\mid F\in L^2([0,R_0])\right\},
\] 
where 
\[ 
  \alpha_0(r) \coloneqq
  \frac{e^{ \frac{1}{2}(\lambda_0+\mu_0)(r)}}{\sqrt{r(\lambda_0'+\mu_0')(r)}},
  \quad r\in ]0,R_0[.
\] 
In addition, if $f(x,v)=\vert\phi'(E)\vert\,w\,\alpha_0(r)\,F(r)$
and $g(x,v)=\vert\phi'(E)\vert\,w\,\alpha_0(r)\,G(r)$, then
\begin{equation}\label{mathur_sp}
  \langle f,g\rangle_H = \langle F,G \rangle_{L^2([0,R_0])};
\end{equation}
a key ingredient here is the identity
\begin{align*} 
  \iint w^2 |\phi'(E)| \, dv =
  \frac{e^{-2\lambda_0(r)-\mu_0(r)}}{4\pi r} (\lambda_0'+\mu_0') (r),\ r>0,
\end{align*}
which follows from \eqref{intpar1} and \eqref{laplusmu}. 
Based on these observations, 
the {\em reduced operator} or {\em Mathur operator}
\[
\M\colon L^2([0,R_0])\to L^2([0,R_0]),\;F\mapsto G
\]
is defined as follows.
First map $F\in L^2([0,R_0])$ to $f\in H^\mathrm{odd}$ defined by
\begin{align*} 
  f(x,v)\coloneqq\vert\phi'(E)\vert\,w\,\alpha_0(r)\,F(r)\quad
  \text{for a.e.\ } (x,v)\in D.
\end{align*}
Next map this $f$ to $Qf\in \Rs (\sqrt\Ro)$. Then there
exists a unique $G\in L^2([0,R_0])$ such that
\begin{align*}
  Q f(x,v)=\vert\phi'(E)\vert\,w\,\alpha_0(r)\,G(r)\quad
  \text{for a.e.\ } (x,v)\in D,
\end{align*}
which completes the construction of the map $\M$.

The relation of $\M$ with $Q$ immediately implies that
$\gamma\neq 0$ is an eigenvalue of $Q$ if and only if it is an eigenvalue of
$\M$, and the multiplicities are equal; concerning the latter notice
that by \eqref{mathur_sp} orthogonality of eigenfunctions is preserved.
By the same relation it is easy to verify that $\M$ inherits the functional
analytic properties of $Q$: $\M$ is a bounded, linear, symmetric,
non-negative, compact operator, cf.~Prop.~\ref{Qprop}.

A draw-back of the construction seems to be that the
Birman-Schwinger operator $Q$ and hence also the Mathur operator $\M$
contain the inverse operator of $\B$. As noted before, this inverse
cannot be given explicitly, which seems to make it unclear how to
apply the machinery above to specific examples. However, the right
inverse $\tilde\B^{-1}$ of $\B$ can actually be given explicitly,
cf.~\cite[Def.~5.7]{GueStrRe}.
The operator $\B \colon \Ds(\T) \cap \Ns(\B)^\bot  \to \Rs(\B)$ is
bijective with bounded inverse given by 
\[ 
\B^{-1} = (\mathrm{id}-\Pi)\widetilde{\B}^{-1},
\]
where $\Pi \colon H \to \Ns(\B)$ is the orthogonal projection
onto the kernel $\Ns(\B)$ of $\B$, see for example \cite[Section~5.4]{HiSi}.
This information turns out to be sufficient to derive an integral
representation of $\M$ which is quite workable in applications.
\begin{proposition}\label{MKdef}
  For  $G\in L^2([0,R_0])$,
  \[ 
  (\M G)(r) =  \int_0^{R_0} K(r,s) G(s) \, ds, \quad r\in[0,R_0],
  \]
  where the kernel $K \in L^2([0,R_0]^2)$ is defined as
  \[ 
  K(r,s) = e^{\frac{1}{2}(\mu_0(r) + 3\lambda_0(r))}
  e^{\frac{1}{2}(\mu_0(s) + 3\lambda_0(s))}
  \frac{\sqrt{2r\mu_0'(r)+1}\sqrt{2s\mu_0'(s)+1}}{rs}  \, I(r,s),
  \] 
  with
  \[ 
  I(r,s) =
  \left \langle (\mathrm{id}-\Pi) \left ( |\phi'|E e^{-\lambda_0-\mu_0}
  \mathbf{1}_{[0,r]} \right ),
  |\phi'|  E e^{-\lambda_0-\mu_0}  \mathbf{1}_{[0,s]}  \right \rangle_H,\
  0\leq r,s \leq R_0 .
  \] 
  The kernel is symmetric, i.e., $K(r,s)=K(s,r)$, and
  $\M$ is a Hilbert-Schmidt operator, see~\cite[Thm.~VI.22 et~seq.]{ReSi1}.
\end{proposition}
If one now combines the relations between the spectra of the Antonov operator
$\L$, the Birman-Schwinger operator $Q$, and the Mathur operator $\M$
with general results on Hilbert-Schmidt operators the following linear
(in)stability information on (EV) results.
\begin{theorem}\label{mathur_stab}
  \begin{itemize}
  \item[(a)]
    The steady state is linearly stable if, and only if,
    \[
    \sup
    \limits_{G\in L^2([0,R_0]),\,\|G\|_{2}=1 }
    \int_{0}^{R_0}\int_{0}^{R_0} K(r,s) G(r)G(s) \, ds dr < 1.
    \]
    If equality holds, there exists a zero-frequency mode but
    no exponentially growing mode.
  \item[(b)]
    The number of exponentially growing modes of the steady state
    is finite and strictly bounded by $\|K\|_{L^2([0,R_0]^2)}^2$. 
  \item[(c)]
    The steady state is linearly stable if $\|K\|_{L^2([0,R_0]^2)}<1$.
  \end{itemize}
\end{theorem}
For a detailed proof we refer to \cite{GueStrRe}. Here we want to discuss an
application of these techniques yielding a result which was not obtained
by the methods in the previous sections, namely, we want to consider the
stability of a shell of Vlasov matter surrounding a Schwarzschild black
hole. To this end we generalize the steady state ansatz to
\eqref{ststansatz_evL}, we place a Schwarzschild singularity of fixed
mass $M>0$ at the center, multiply the ansatz for the particle distribution
$f$ with a parameter $\delta > 0$ and keep the condition
\eqref{mainstabcond}. One can show that there exist corresponding
steady states of (EV) where the Vlasov shell has finite mass, finite extension,
and is of course situated outside the Schwarzschild radius of the black hole;
for the details of the construction of these steady states we refer to
\cite[Section~2.2]{GueStrRe}.
If one keeps the ansatz function with its cut-off energy and
cut-off angular momentum and the mass $M$ of the Schwarzschild singularity
fixed one can show that for $\delta>0$ sufficiently
small, the corresponding effective potential for the particle motion
still has a single-well structure in the sense of Lemma~\ref{sphsymmchar}
and Section~\ref{stst_nonrel}. This allows the introduction of action-angle
variables for the stationary characteristic flow, which was the tool behind
many of the constructions in the present section, so that these
constructions continue to function also for the case with a central
black hole, provided the mass of the black hole dominates the mass in
the surrounding Vlasov shell. If one applies these constructions,
one obtains the following result.
\begin{theorem}\label{shell_stab}
  There exist families of steady states
  $(f^\delta,\lambda^\delta,\mu^\delta)_{\delta>0}$
  of (EV) with a Schwarzschild singularity of mass $M>0$ at the center
  surrounded by a shell of Vlasov matter with particle distribution
  $f^\delta$, where the parameter
  $\delta>0$ controls the size of the Vlasov shell.
  These steady states are linearly stable for $\delta>0$ sufficiently
  small.
  For $\delta\to 0$ the metric converges to the vacuum Schwarzschild
  metric of mass $M$, uniformly on $]2M,\infty[$,
  and the density $f^\delta$ converges to zero pointwise. 
\end{theorem}
One should note that the characteristic flow for the particles in the
shell of Vlasov matter is very different from the flow induced by
null geodesics, which governs the propagation of massless particles
and perturbations of the metric. The result in Theorem~\ref{shell_stab}
is very different from the result in \cite{DaHoRoTa}.
\section{Numerical observations, conjectures, and open problems}
\label{numetc}
\setcounter{equation}{0}
Maybe the most important single fact about the stability problem
for the Einstein-Vlasov system is that along a one-parameter family
$(f_\kappa,\lambda_\kappa,\mu_\kappa)_{\kappa>0}$ of steady states with
some microscopic equation of state $\varphi$ with $\varphi' > 0$,
cf.~Prop.~\ref{ssfamilies}, the steady states change
from being stable
to being unstable when the central redshift $\kappa$ changes from
being small to being large. This is a genuinely relativistic feature
which has no parallel for the non-relativistic Vlasov-Poisson system.

We discussed some first steps towards understanding
this $\kappa$-dependence of the stability behavior
on the linearized level in
Sections~\ref{stabkappasmall},~\ref{stabkappalarge},~\ref{FA},
and there is ample numerical evidence that this behavior is
very general for (EV)
and is true on the nonlinear level,
cf.~\cite{AnRe1,Gue_e_a,GueStRe21}. Hence any
successful, comprehensive stability analysis for (EV) will have to take this
phenomenon properly into account.

One key step towards understanding this behavior would obviously be to
find a criterion for exactly when the change from stability to instability
occurs. For the Einstein-Euler system the turning-point principle clearly
specifies the points along the so-called mass-radius curve of a one-parameter
steady state family, where stability changes to instability or the other way,
cf.~\cite{HaThWaWh,Mak,ScWa2014}; in \cite{HaLin}
{\sc Had\v zi\'c} and {\sc Lin} give a rigorous proof
of the turning-point principle for the Einstein-Euler system.
But with the Vlasov matter model
instead of a compressible, ideal fluid numerical evidence shows
the analogous turning-point principle
to be false \cite{Gue_e_a}.
This issue has also been discussed in the astrophysics literature
\cite{AbEtal1994, Fa1970, Ip69b,Ip1980,IT68, RaShTe1989_2,
  ShTe1985,ZeNo, ZePo},
where the behavior of the so-called binding energy has been suggested
as an alternative stability indicator.
The (fractional) binding energy of a steady state
$(f_\kappa,\lambda_\kappa,\mu_\kappa)$ is defined as
\[  
E_{b,\kappa} = \frac{N_\kappa-M_\kappa}{N_\kappa},
\] 
where
\[
M_\kappa =\iint f_\kappa \pv dv\,dx,\quad
N_\kappa =\iint e^{\lambda_\kappa} f_\kappa\, dv\,dx
\]
are its ADM-mass and particle number. One can distinguish two forms
of the binding energy hypothesis.
The \textit{weak binding energy hypothesis} claims that steady states
are stable at least up to the first local maximum of the binding energy
curve parameterized by the redshift.
The \textit{strong binding energy hypothesis} claims that steady states
are stable precisely up to the first local maximum of the binding energy
curve and become unstable beyond this maximum. In \cite{GueStRe21}
numerical evidence against the strong binding energy hypothesis
is given and it is shown that along the binding energy curve
several stability changes can occur. The question from which quantity
one can predict the stability behavior of the corresponding
steady state is open even on the level of numerical simulations,
and a good candidate could indicate how to make progress
on the rigorous analysis of the stability issue.

A further question which has been investigated numerically in
\cite{AnRe1,Gue_e_a} is how a stable or an unstable steady state
reacts to perturbation. Upon perturbation, a stable steady state
typically starts to oscillate with an undamped or damped amplitude,
very similarly to what we discussed for the (VP) case. The
reaction of an unstable steady state to perturbations is much more
interesting. Depending on the ``direction'' of the perturbation
it collapses to a black hole or it seems to follow
some sort of heteroclinic orbit to a different, stable steady state about
which (the bulk of) it starts to oscillate. Steady states with a
very large central redshift
may upon perturbation also disperse towards flat Minkowski space
instead of following a heteroclinic orbit as described above.

Obviously, there are in this context plenty of challenging questions
awaiting rigorous mathematical analysis. We emphasize that the
stability question for (EV) has also received a lot of attention
in the astrophysics literature; in addition to the citations above
we mention \cite{FaIpTh,FaSu1976a,FaSu1976b,KaHoKl1975,Th1966}.

One aim certainly must be to prove nonlinear stability
of steady states with small redshift, i.e., for steady states where
linear stability holds according to Theorem~\ref{linstabthm} and the results
in Section~\ref{FA}. As we explained in Section~\ref{ss_criticality}
it seems doubtful whether the global variational approach based
on the energy-Casimir functional $\H_C$ can succeed in the
(EV) case, although it is very
successful for the (VP) case as we saw in Section~\ref{globvar_vp}.
Probably a better chance for generalizing
it to the (EV) case exists for the local minimizer approach discussed
in Section~\ref{locvar_vp} for the (RVP) case. The situation may improve 
if one considers a suitable reduced functional derived from $\H_C$,
such as we discussed for (VP) in Section~\ref{globvar_vp}, cf.~\eqref{hcrdef}.
In \cite{Wo} such a functional was derived for (EV), but the approach
there suffers from two defects. Firstly, some of the arguments in
\cite{Wo} are wrong; the main assertions in \cite{Wo} have not been proven,
cf.~\cite{AnKu20}. Secondly, even if correct the results in
\cite{Wo} would not imply any stability assertion since the
existence of minimizers to the reduced functional relies on certain
barrier conditions which are not known to be respected by the time-dependent
solutions. This second, conceptual problem persists even though in
\cite{AnKu} the authors were able to rigorously prove some of the assertions
in \cite{Wo}. In any case, nonlinear stability for (EV) is open, and we
believe that new types of (conserved) functionals, probably
involving derivatives of the metric coefficients, or/and new
types of barrier conditions which are respected by time-dependent 
solutions are needed.

A second aim should be to prove nonlinear instability in situations
where the existence of an exponentially growing mode is known by
Theorem~\ref{gromo_thm}. It seems inconceivable that
an exponentially growing mode exists and the steady state is nonlinearly
stable anyway, but saying this is no proof. We point out that in
the (VP) plasma physics case the step from an exponentially
growing mode to nonlinear
instability has been made in \cite{GStr}, see also \cite{FrStrVi}.
We believe that it is a non-trivial and worthwhile project to
prove the analogous result for (EV), even though the outcome will
probably not be surprising; notice that such unstable states
should upon proper perturbation collapse to a black hole,
and initial data which lead to the formation of black holes
are very important in themselves.
An interesting aspect here is that for the gravitational (VP) case
the existence of exponentially growing modes has so far not
been rigorously proven for potentially unstable steady
states---those with sufficiently non-monotone microscopic
equation of state $\varphi$---but we refer to \cite{WaGuoea} for a numerical
construction, see also \cite{GuoLin}.

To conclude this section we point out that the stability problem reviewed
above also poses some open problems which refer to the structure of
the steady states themselves, but which have some bearing on the dynamic
stability problem. It would be interesting to know for which steady states
the effective potential has a single-well structure and allows for the
introduction of action-angle variables for the stationary characteristic
flow. Numerical evidence seems to suggest that this is true for
isotropic steady states, but not necessarily for general ones.
A proof exists only for isotropic steady states which satisfy the
condition
$\sup \frac{2 m(r)}{r} \leq \frac{1}{3}$, while numerical evidence
suggests that this Buchdahl quotient is bounded by $\frac{1}{2}$,
but whether this improved Buchdahl bound indeed holds for all isotropic
(EV) steady states and whether it implies a single-well structure is
open.
\section{Strict, global energy minimizers need not be stable}
\label{theexample}
\setcounter{equation}{0}
Consider the Hilbert space
\[
H \coloneqq \left\{(z_k)_{k\in \N} \mid
z_k = (x_k,p_k)\in \R^2,\ k\in \N,\ \mbox{and}\
\sum_{k=1}^\infty (x_k^2 + p_k^2) < \infty \right\}
\]
equipped with the norm
\[
\| z\| \coloneqq \left(\sum_{k=1}^\infty (x_k^2 + p_k^2)\right)^{1/2},\
z=(z_k)_{k\in \N} = ((x_k,p_k))_{k\in \N}.
\]
On this space we define a linear dynamical system via
\[
\dot x_k = p_k,\ \dot p_k = -\frac{1}{k^2}\, x_k,
\ \mbox{i.e.},\ \ddot x_k =  -\frac{1}{k^2} x_k,\ k\in \N,
\]
which can be written as
\begin{equation}\label{exsys}
  \dot z = \L z
\end{equation}
with the linear, bounded operator
\[
\L \colon H \to H,\ \L z \coloneqq ((p_k,-\frac{1}{k^2} x_k))_{k\in \N}.
\]
The operator $\L$ generates a uniform $C_0$ group
$(e^{t\L})_{t\in \R}$ of bounded operators on $H$.
The energy functional
\[
\H \colon H \to \R,\
\H(z)\coloneqq \sum_{k=1}^\infty \frac{1}{2}
\left(\frac{1}{k^2}x_k^2 + p_k^2\right)
\]
is Fr\'{e}chet differentiable with
\[
\langle D\H(z),\delta z\rangle =
\sum_{k=1}^\infty \left(\frac{1}{k^2}x_k \,\delta x_k + p_k \,\delta p_k\right),
\]
and $\H$ is a conserved quantity: $\frac{d}{dt} \H (e^{t\L} z) =0$
for any $z\in H$. Since the system is linear, $0$ is a stationary solution,
and it is the unique, strict minimizer of the energy $\H$. However,
$0$ is dynamically unstable in the sense of Lyapunov. To see this, we
choose $\epsilon =1$ and let $\delta>0$ be arbitrary.
Fix some $n\in \N$ such that $\frac{1}{2} n \delta > 1$ and let
$\mathring z \coloneqq \left(\delta_{nk}(0,\frac{\delta}{2})\right)_{k\in \N}
\in H$.
The solution with these initial data is given by
\[
z_k (t) \coloneqq \delta_{nk}\frac{\delta}{2}
\left(-k \cos\left(\frac{t}{k}+\frac{\pi}{2}\right),
\sin\left(\frac{t}{k}+\frac{\pi}{2}\right)\right),\ k\in \N,
\]
and $\|\mathring z\| = \frac{\delta}{2} < \delta$, while
$\|z(\frac{3\pi}{2} n)\| = |z_n(\frac{3\pi}{2} n)| =\frac{1}{2} n \delta > 1$,
which shows that the steady state is unstable.

The nice thing about this example, which in some form or other is certainly
known and is obvious enough, is that the instability is not triggered
by some nonlinear correction to the linear(ized) dynamics, but solely by
the infinitely many directions in which a solution can escape.

It is also obvious that there is no compactness along minimizing sequences
of $\H$, such as we exploited in Section~\ref{globvar_vp}. Let
$z^n\coloneqq((\delta_{nk},0))_{k\in \N}$. Then $(z^n)_{n\in \N}$
is a minimizing sequence of $\H$, $\H(z^n) = \frac{1}{2 n^2} \to 0$,
but it converges no better than weakly.

Finally, the operator has the spectrum
$\sigma(\L) = \{\pm \frac{i}{k} \mid k\in \N\}$
consisting only of isolated eigenvalues of multiplicity $1$. When viewed as
a second order system, \eqref{exsys} takes the form
\[
\ddot x = \widetilde{\L} x,
\]
where the bounded, self-adjoint operator $\widetilde{\L}\colon l^2 \to l^2$
is defined by
\[
\widetilde{\L} x \coloneqq\left(-\frac{1}{k^2}x_k\right)_{k\in \N}
\]
and has spectrum
$\sigma(\widetilde{\L}) = \{-\frac{1}{k^2} \mid k\in \N\}$.
One should compare this with the situation for (VP) or (EV)
where we also had a first and a second order version of the linearized
system with a self-adjoint operator governing the latter.

\end{document}